\documentclass[runningheads]{llncs}
\usepackage[T1]{fontenc}

\usepackage[colorlinks=true, allcolors=blue]{hyperref}
\usepackage[english]{babel}
\usepackage{setspace}
\usepackage{xkeyval}
\usepackage{caption}

\usepackage{authblk}

\usepackage{amsmath}
\usepackage{amssymb}
\usepackage{amsfonts}
\usepackage{mathtools}
\usepackage{enumerate}
\usepackage{bbm}

\allowdisplaybreaks
\usepackage[noend, linesnumbered]{algorithm2e}
\RestyleAlgo{ruled}
\SetKwInput{Input}{Input}
\SetKwInput{Output}{Output}
\SetKw{Return}{Return}
\SetKw{Continue}{Continue}
\SetKw{Break}{Break}
\DontPrintSemicolon

\spnewtheorem{observation}[theorem]{Observation}{\bfseries}{\itshape}

\DeclareMathOperator*{\argmin}{arg\,min}

\newcommand{\EE}{\mathbb{E}}
\newcommand{\poly}{\textnormal{poly}}

\newcommand{\maxMarginal}[1][n]{#1}
\newcommand{\maxValue}[1][n]{\binom{#1}{2}}

\newcommand{\marginal}[2]{\deg_{#2}(#1)}
\newcommand{\pseudoMarginal}[4][]{\frac{\marginal{#2}{#3} \cdot \left(\marginal{#2}{#3} + 1 \right)}{#4_{#2#1}}}
\newcommand{\inversePseudoMarginal}[4][]{\frac{#4_{#2#1}}{\marginal{#2}{#3} \cdot \left( \marginal{#2}{#3} + 1  \right)}}

\newcommand{\estMarginal}[1]{\widetilde{d_{#1}}}

\newcommand{\constPseudoCore}{\ensuremath
\mathchoice
    {\frac{\varepsilon^4}{36}}
    {\varepsilon^4 / 36}
    {\varepsilon^4 / 36}
    {\varepsilon^4 / 36}
}
\newcommand{\constPseudoCoreInternal}{\ensuremath
\mathchoice
    {\frac{\varepsilon}{6}}
    {\varepsilon / 6}
    {\varepsilon / 6}
    {\varepsilon / 6}
}
\newcommand{\constInverseDegreeBound}{\ensuremath
\mathchoice
    {\frac{36 e}{\left(1-\varepsilon\right)^2 \cdot \varepsilon^5 \cdot \ln(2)}}
    {(36 e) / (\left(1-\varepsilon\right)^2 \cdot \varepsilon^5 \cdot \ln(2))}
    {(36 e) / (\left(1-\varepsilon\right)^2 \cdot \varepsilon^5 \cdot \ln(2))}
    {(36 e) / (\left(1-\varepsilon\right)^2 \cdot \varepsilon^5 \cdot \ln(2))}
}
\newcommand{\constConcentration}{\ensuremath
\mathchoice
    {\frac{2}{\varepsilon^2}}
    {2 / \varepsilon^2}
    {2 / \varepsilon^2}
    {2 / \varepsilon^2}
}

\newcommand{\funcPseudoCore}{\ensuremath \gamma
\left( n^2 + \mathchoice
{\frac{6}{\varepsilon}}
{(6/\varepsilon)}
{(6/\varepsilon)}
{(6/\varepsilon)}
n\right)}
\newcommand{\funcPseudoCoreInternal}{\ensuremath \gamma' n}
\newcommand{\funcInverseDegreeBound}[1][]{\ensuremath \delta_{#1} \cdot \mathchoice
    {\frac{\ln\ln(n)}{n}}
    {\ln\ln(n) / n}
    {\ln\ln(n) / n}
    {\ln\ln(n) / n}
}

\newcommand{\funcConcentration}{\ensuremath \kappa \cdot n^{1/2} \cdot \ln\left(2n^2\right)}
\newcommand{\funcConcentrationPlus}{\ensuremath n^{1/2} \left( \kappa \cdot \ln\left(2n^2\right) + 1 \right)}

\title{Linear Contracts for Supermodular Functions Based on Graphs}

\author{Kanstantsin Pashkovich\and Jacob Skitsko}%

 \institute{Deptartment of Combinatorics and Optimization \\ University of Waterloo, Canada\\\email{kpashkovich@uwaterloo.ca} \quad\email{jskitsko@uwaterloo.ca}
 }

\begin{document}
\date{}
\maketitle

We study linear contracts for combinatorial problems in multi-agent settings. In this problem, a principal designs a linear contract with several agents, each of whom can decide to take a costly action or not. The principal observes only the outcome of the agents' collective actions, not the actions themselves, and obtains a reward from this outcome. Agents that take an action incur a cost, and so naturally agents require a fraction of the principal's reward as an incentive for taking their action. The principal needs to decide what fraction of their reward to give to each agent so that the principal's expected utility is maximized.

Our focus is on the case when the agents are vertices in a graph and the principal's reward corresponds to the number of edges between agents who take their costly action. This case represents the natural scenario when an action of each agent complements actions of other agents though collaborations. Recently, \cite{deo2024supermodular} showed that for this problem it is impossible to provide any finite multiplicative approximation or additive FPTAS unless $\mathcal{P} = \mathcal{NP}$. On a positive note, the authors provided an additive PTAS for the case when all agents have the same cost. They asked whether an additive PTAS can be obtained for the general case, i.e when agents potentially have different costs. We answer this open question in positive.

\section{Introduction}

 In 2016, the Nobel Prize in economics was awarded to Oliver Hart and Bengt Holmstr\"om for their work on contract theory~\cite{Grossman1992,2ffe85e7-2b99-303d-a244-cf884ddb8386}. Today contracts play one of central roles in economic theory and in our everyday life.  "Modern economies are held together by innumerable contracts", {\it The Royal Swedish Academy of Sciences.}~\cite{nobel-contract}.

In recent years, contract theory has garnered an increasing interest in the theoretical computer science and optimization communities. One factor for this growing interest is the clear connection between contracts and fundamental optimization problems. The arising contract design problems turn out to be hard even in relatively structured settings, e.g. it is $\mathcal{NP}$-hard to find the optimal contract for additive reward functions in a multi-agent setting \cite{dutting2023multi}. So, it is natural to ask whether it is possible to efficiently find a contract that achieves \emph{approximately} optimal expected utility. By today, there are various results about approximations and hardness of approximation in the contract design theory

Recently, \cite{deo2024supermodular} studied multi-agent contract design when the reward function is supermodular and is based on a graph, with all agents having the same cost.  This class of functions represents the situation when the agents complement the actions of one another through collaboration. In other words the given graph represents potential collaborations between agents, and each active agent increases the total reward by the number of collaborations between this agent and other active agents. 

They showed that it is impossible to provide a finite multiplicative approximation or any additive FPTAS, unless $\mathcal{P} = \mathcal{NP}$. On a positive note, they develop an additive PTAS when all agents have the same costs. Their PTAS built on the techniques developed in~\cite{arora1995polynomial}, and established connections to densest subgraph problems. In~\cite{deo2024supermodular}, it was left as an open question whether it is possible to  provide a PTAS not only in the case when all agents have the same cost, but also for the general case. In this paper, we provide an additive PTAS for the graph based supermodular functions considered by~\cite{deo2024supermodular} for the case with general agent costs.

Apart from resolving the open question stated in~\cite{deo2024supermodular} and obtaining new tools specific to this question, we see our contribution also in developing techniques to round solutions for a certain class of linear programs. When coupled with randomized rounding, our techniques ensure that the value of the objective function and feasibility of constraints deteriorate only slightly, see Theorem~\ref{thm:main_supermodular2} below. This class of linear programs captures linear programs similar in spirit to cover problems where both the "ground set elements" and "covering sets" are identified with the same set. In particular, both are identified with agents. In this case, each linear program requires that each agent that participates in a "cover" is also "covered" a sufficiently high number of times, see linear programs formed by the constraints~\eqref{lp_intro:high_degree},~\eqref{lp_intro:low_degree},~\eqref{lp_intro:bounds}.

\subsection{Related Work}

Combinatorial functions naturally arise in contract design theory whenever there are several agents. Consider the setting where each agent decides to exert an effort or not, then the expected reward can be seen as a function defined on subsets of agents. Our work is most related to the work of \cite{dutting2023multi,deo2024supermodular}. They both consider the same multi-agent setting that we do, with related combinatorial functions. \cite{dutting2023multi} considers general XOS and submodular functions, while \cite{deo2024supermodular} considers a particular class of  supermodular function based on graphs. They provide approximation algorithms and demonstrate hardness of approximation, which justifies the study of approximately optimal contracts. \cite{dutting2023multi} gives a constant approximation for XOS and submodular functions, and shows that there is also a constant such that it is $\mathcal{NP}$-hard to approximate the optimal expected utility within that constant. \cite{deo2024supermodular} proved that for supermodular functions based on graphs the problem admits no multiplicative approximation or additive FPTAS unless $\mathcal{P} = \mathcal{NP}$. On a positive note, \cite{deo2024supermodular} showed how to design an additive PTAS for the case when agents' costs are identical. Earlier, \cite{babaioff2006combinatorial} considered combinatorial functions in the multi-agent setting and gave an algorithm to compute the optimal contract for AND networks. It was then observed by \cite{emek2012computing} that computing the optimal contract for OR networks was $\mathcal{NP}$-hard. Note, the functions corresponding to OR networks are submodular. Also, several works show that there are natural limitations for such approximations and these limitations are related to both the computational hardness  \cite{dutting2019simple,dutting2021complexity,castiglioni2021bayesian,castiglioni2022designing} and to the query complexity~\cite{ezra2024approximability,dutting2024query}.

A setting with reward function based on edges in a graph was also considered in \cite{dasaratha2023equity}.
They focus on characterizing the optimal allocation in their setting.
There is an alternative multi-agent contract model, studied in \cite{castiglioni2023multi}. In this model each agent determines their own individual outcome, which removes externalities. Under mild regularity conditions, they show the optimal contract in their setting for supermodular functions can be computed in polynomial time, and for submodular functions can be approximated within~$(1 - 1/e)$. There are related models with a single agent but multiple actions, where the agent can take a subset of actions. It is also natural to consider combinatorial functions in such a setting. \cite{dutting2022combinatorial}  calculated the optimal contract for gross-substitutes functions, and show the hardness for submodular functions. The techniques used were analyzed more closely in \cite{dutting2024combinatorial}. One could also consider a setting where the agent is able to take actions sequentially, after observing the outcomes of previous actions \cite{ezra2024sequential,hoefer2024contract}. This has a natural connection to Pandora's box problems \cite{weitzman1978optimal}. In concurrent work, \cite{duetting2024multi} considered a setting with multiple agents each able to choose a subset of actions. For submodular reward functions in this setting, \cite{duetting2024multi} obtained a constant approximation with value and demand oracle access and showed that no PTAS exists. Finally, one can consider settings with multiple principals \cite{alon2024incomplete}.

An expansive number of different settings for contract theory have been explored in recent years, as the setting has recently caught the eye of many researchers in the theoretical computer science community. As a small list, other recent models have included contracts with inspections \cite{ezra2024contracts}, ambiguous contracts \cite{dutting2023ambiguous}, agents choosing from menus of contracts \cite{guruganesh2034contracts,bernasconi2024agent}, contracts with a learning agent \cite{guruganesh2024contracting}, contracts for delegating machine learning tasks \cite{saig2024delegated,ananthakrishnan2024delegating} or selling data to a machine learner \cite{chen2022selling}, hidden but signaled outcomes \cite{babichenko2022information}, and Bayesian agent types \cite{alon2021contracts,alon2022bayesian,xiao2020optimal,guruganesh2021contracts,guruganesh2034contracts,castiglioni2022designing,castiglioni2024reduction}.

Our results on the supermodular functions also have close ties to approximations for dense problems. \cite{arora1995polynomial} gave an additive PTAS for many instances of dense problems, which uses exhaustive sampling to approximately solve a polynomial program. Similar sampling techniques were used in \cite{daskalakis2011oblivious,barman2015approximating} to give a PTAS for finding small probability mixed Nash Equilibrium. 

\subsection{The Contract Design Setting}

We have a single principal who interacts with a set of $n$ agents ${[n]} = \{1, \dots, n\}$. We focus on the setting with binary actions where each agent either exerts an effort or not. An agent $i$ exerting an effort incurs a cost $c_i \in \mathbb R_{\geq 0}$ for the agent, and otherwise they incur no cost. 

We are given a function $f: 2^{{[n]}} \rightarrow \mathbb R_{\geq 0}$ representing the expected reward of the principal when a group of agents decides to exert an effort. 
The function $f$ is monotone, i.e. $f(S) \leq f(T)$ holds whenever $S \subseteq T$. The function $f$ is also normalized so that $f(\varnothing) = 0$ and $f([n]) \leq 1$. We define the \emph{marginal value} of agent $i$ with respect to the set $S$ as $f(i \mid S) := f(S \cup \{i\}) - f(S)$.

A key challenge in contract theory is that the principal obtains an outcome and the corresponding reward, but cannot observe the actions of the agents. Thus,  the principal has to provide incentives for agents to take costly actions, and these incentives have to be based solely on the outcome. 

To achieve this, the principal has to come up with a contract. We restrict focus to \emph{linear contracts}, which are popular in the literature. A linear contract is defined by  a single  vector $t=(t_1,\ldots, t_n)^T$ and the agent $i$, $i\in [n]$ always receives a positive transfer $r \cdot t_i$ whenever the principal obtains the reward $r$.

Let us assume that the agents in $S$ exert an effort while all agents not in $S$ do not. Then, for a linear contract $t=(t_1,\ldots, t_n)^T$  the \emph{expected utility} of an agent $i$, $i\in[n]$ equals $f(S) \cdot t_i - c_i$ if $i$ is in $S$; and equals $f(S) \cdot t_i$ if $i$ is not in $S$. Moreover, the principal's expected utility is
\[
    f(S)\cdot \left( 1 - \sum_{i \in {[n]}} t_i \right)\,.
\]
Observe that the actions of agents can affect the expected utility of other agents. We assume that each agent attempts to maximize their expected utility. Following existing contract theory literature \cite{babaioff2006combinatorial,dutting2023multi,deo2024supermodular}, we see that the set $S$, $S\subseteq [n]$ forms a pure Nash equilibrium  if and only if
$f(S)\cdot t_i - c_i
    \geq f(S\setminus i)\cdot t_i
    $ for all  $i\in S$ and 
    $f(S)\cdot t_i
    \geq f(S \cup i)\cdot t_i - c_i$
    for all  $i\notin S$.
Then a linear contract $t^*=(t^*_1,\ldots, t^*_n)^T$ that incentivizes exactly the agents in $S$ to exert an effort while maximizing the principal's expected utility satisfies
$t^*_i = c_i/f(i\mid S)$ for all  $i\in S$ and $t^*_i = 0$ for all $ i \notin S$. Here, $c_i/f(i\mid S)$ is $0$ if $c_i=0$, but otherwise infinity when $f(i\mid S) = 0$. Using the contract $t^*$ the principal's expected utility is given by the function
 \begin{equation}\label{eq:function_g}
      g(S) := \left( 1 - \sum_{i\in S} \frac{c_i}{f(i\mid S)} \right) f(S) \,.
 \end{equation}
  Hence, whenever the principal decides to incentivize  the set of agents $S$, the maximum principal's expected utility  equals $g(S)$. Now, our goal is to determine  $S$, $S \subseteq [n]$ that maximizes $g(S)$.

We are given a graph $G=(V,E)$ with $|V|=n$, and the agents are vertices in $V$. We consider the reward function $f(S) = |E(S)| / \binom{n}{2}$. For the sake of exposition, we modify the costs $c_i$, $i\in V$ as follows $c_i \gets c_i \cdot \maxValue$. Our goal is thus to find a set $S\subseteq V$ that maximizes the following function, which corresponds to the function in~\eqref{eq:function_g}
\[
    g(S) =   \left(1 - \sum_{i \in S} \frac{c_i}{\marginal{i}{S}}\right) \frac{|E(S)|}{\maxValue} \,,
\]
where $E(S):=\{uv\in E\,:\, u \in S,\, v\in S\}$ and $\deg_S(i):=|N(i)\cap S|$. We further define functions for the left and right factors of $g(S)$
\[
    L(S) := 1 - \sum_{i \in S} \frac{c_i}{\marginal{i}{S}} \quad \text{and} \quad
    R(S) := \frac{|E(S)|}{\maxValue} \,.
\]

For linear contracts in the multi-agent setting where the reward function is a supermodular function based on graphs,  \cite{deo2024supermodular} showed that it is hard to obtain any multiplicative approximation unless $\mathcal{P} = \mathcal{NP}$. \cite{deo2024supermodular} left as an open question whether there is an additive PTAS. In this paper we answer this question positively.

\begin{theorem}\label{thm:intro_supermodular}
    Let us consider linear contracts in the multi-agent setting, where the reward function is a supermodular function based on graphs. There exists an additive PTAS in the case of general costs.
\end{theorem}

Before providing an overview and intuition behind our PTAS, let us introduce some notation. We denote by $OPT$ the maximum value of $g(S)$ over $S$, $S\subseteq V$, and we let $S^*$ be the corresponding optimal set of agents.
We assume that $OPT$ is at least $\varepsilon$, since otherwise the task of additive approximation up to $\varepsilon$ is trivial, by taking~$\varnothing$.

Let us define parameters $\kappa$, $\sigma$, $\delta$, $\delta_\varepsilon$, $\gamma'$, $\gamma$ that do not depend on $n$. For now, it suffices to ignore their formal definitions and then return to them for verifying technical details and proofs. Here, $\kappa$ is needed for concentration inequalities, $\sigma$ is used for sampling bounds, $\delta$ and $\delta_\varepsilon$ are used for degree estimates, and $\gamma', \gamma$ are parameters for working with "pseudo-cores".
\begin{gather*}
    \kappa := \constConcentration, \quad
    \sigma:=\varepsilon^6, \quad
    \delta := \constInverseDegreeBound, \quad 
    \delta_\varepsilon := \frac{\delta}{\varepsilon}, \quad
    \gamma' := \constPseudoCoreInternal, \quad
    \gamma := \constPseudoCore \,.
\end{gather*}

\section{Overview and Intuition}

In this section, we provide overview and intuition behind the PTAS developed to prove Theorem~\ref{thm:intro_supermodular}. We would like to emphasize that the exposition in this section is done with the omission of technical details. The sole purpose of this section is to provide intuition for the objects and procedures introduced in later sections. No part of this section should be understood as a formal definition.

Let $S^*$ be an optimal set, i.e.\ the set $S^*$, $S^*\subseteq V$ that achieves $\max_{S\subseteq V} g(S)$. Our approach roughly attempts to obtain enough information about $S^*$ so that we can set up and round an appropriate linear program. More truthfully, we show the existence of a structured near optimal set $S'$. We cannot efficiently find $S'$, but we can obtain enough of information about $S'$ to set up our LP. We then round the LP optimal solution to a near optimal solution $S''$ for our original problem.

\textbf{Existence of a Structured Near Optimal Solution.}
As a first step we generalize the approach built on cores from~\cite{deo2024supermodular}. A core of $S^*$ no longer always suffices as a structured near optimal solution. We consider a \emph{pseudo-core} of $S^*$ instead. In particular, we show that there exists a set $S'$, $S'\subseteq S^*$ such that:
\begin{enumerate}[(i)]
    \item the value $g(S')$ of $S'$ is not far from the optimal value $g(S^*)=OPT$, in particular $g(S')\geq g(S^*)-\varepsilon$;
    \item the degree $\deg_{S'}(v)$ of each vertex $v$ in $S'$ is sufficiently high in comparison to its cost $c_v$, in particular for each vertex $v\in S'$ we have
    \begin{equation}\label{eq:overview_degree_bound}
            \frac{c_v}{\deg_{S'}(v)(\deg_{S'}(v)+1)}\leq \frac{1}{\funcPseudoCore}\,.
    \end{equation}
\end{enumerate}
We remark that~\eqref{eq:overview_degree_bound} implies that  $(\deg_{S'}(v))^2 / c_v = c_v\cdot \Omega(n^2)$ holds for each $v \in S'$ (up to a factor depending on $\varepsilon$).

\textbf{Obtain Information About a Structured Near Optimal Set.}
After showing the existence of a structured near optimal solution $S'$, we focus on obtaining another near optimal solution computationally efficiently. We do not know $S'$ and do not intend on finding it. Nonetheless, we can obtain useful information about $S'$ which lets us later efficiently obtain another near optimal solution $S''$.

Following the idea from~\cite{deo2024supermodular}, we assume that we \textbf{know} the vertices in $V$ with sufficiently high number of neighbours in $S'$. In particular we \textbf{know} the set
\[
    H := \{v\in V\,:\, \deg_{S'}(v) \geq \sigma n\} \,,
\]
and for each $v\in H$ we \textbf{know} $\deg_{S'}(v)$. 

We can assume that we know the set $H$ and the degrees $\deg_{S'}(v)$, $v\in H$ through sampling~\cite{deo2024supermodular,daskalakis2011oblivious}. Sampling lets us obtain accurate estimates for $\deg_{S'}(v)$, $v\in V$ when $\deg_{S'}(v)$ is sufficiently high. We remark that the sampling procedure "guesses" the degrees, and repeats its guesses a polynomial number of times. With high probability one set of guesses is correct. We simply run the following procedures with each guess. From now on, we also assume that we \textbf{know} $|E(S'\cap H)|$. We can exhaustively guess $|E(S'\cap H)|$ as it takes some value in $\{1, \dots, n(n-1)/2\}$, and run the following procedures using each guess. At the very end, we take the best final solution $S''$ from these runs.

Since the vertices in $V \setminus H$ have low degree, we observe that ignoring the vertices outside of $H$ makes very little difference to $R(\cdot)$, more concretely
\[
 R(S^*)\geq R(S'\cap H)\geq R(S^*)-2\varepsilon/3\,.
\]

\textbf{First Attempt at a Linear Program.}
At this point, one can try to set up a linear program similar to the one in~\cite{deo2024supermodular}. Here is a rough idea of how such a linear program could look. The variable $x_v$, $v\in V$ represents whether an agent $v\in V$ is contained in the set $S'$.
  \begin{align}
      \min_{\{x_v\}_{v\in V}} \quad& \sum_{v\in H} \frac{c_v}{ \deg_{S'}(v)}  \cdot x_v \label{eq:overview_first_obj}\\
    \text{subject to} \quad& \sum_{v\in H}\deg_{S'}(v)\cdot x_v \ge 2\cdot |E(S'\cap H)|\label{eq:overview_first_const1}\\
       & \sum_{u \in N(v)} x_u \geq \deg_{S'}(v) &&\text{  for all } v\in H\label{eq:overview_first_const2}\\
        & 0 \leq x_v\leq 1 &&\text{  for all } v\in V\,.\label{eq:overview_first_const3} 
\end{align}
The objective function~\eqref{eq:overview_first_obj} attempts to minimize the negative effect of agents in $H$ on the function $L(\cdot)$. Note agents in $V \setminus H$ are omitted here, since we do not know $\deg_{S'}(v)$ for $v \in V \setminus H$. The constraint~\eqref{eq:overview_first_const1} attempts to guarantee a sufficiently high value of $R(\cdot)$. Constraint~\eqref{eq:overview_first_const1} ignores vertices outside of $H$, but as mentioned above this makes very little difference to $R(\cdot)$. The constraint in~\eqref{eq:overview_first_const2} guarantees that the degrees are behaving properly for agents in $H$. 

In the case when all costs are equal, one can also add the constraints $x_v=0$ for all $v\in V\setminus H$, see~\cite{deo2024supermodular}. Indeed, in that case a core of $S^*$ was a near optimal solution. However, in the general case, we need to be more careful with agents in $V\setminus H$ and cannot omit them from near optimal solutions. Further observe that constraint~\eqref{eq:overview_first_const2} potentially relies on neighbours in $V\setminus H$, and the objective function~\eqref{eq:overview_first_obj} does not account for the agents in $V\setminus H$. This becomes an issue, since in general it is not feasible to efficiently obtain accurate estimates for $\deg_{S'}(v)$, $v\in V\setminus H$. So, we cannot simply add constraints analogous to the constraint~\eqref{eq:overview_first_const2} for $v\in V\setminus H$, or add terms in the objective function for these agents.

\textbf{Final Linear Program.}
Now, we take a closer look at the agents in $V\setminus H$ and $H$. We will split all agents into four different groups $A$, $B$, $C$ and $D$ depending on their cost and on their presence in $H$. Intuitively this gives us the following information on whether or not to include these agents in our solution:
\begin{itemize}
    \item Agents in $A$ might be included; they have high degrees and reasonable costs.
    \item Agents in $B$ might be included; they have middling degrees and costs.
    \item Agents in $C$ should be included; they have low costs.
    \item Agents in $D$ should not be included; they have low degrees and/or high costs.
\end{itemize}

Now we define $A$, $B$, $C$ and $D$ depending on agent cost and the set of high degree agents $H$. Recall at this point we know the costs, the set $H$, and for each agent $v\in H$ we know $\deg_{S'}(v)$.
An agent $v\in H$ lies in $A$ if the agent satisfies the inequality~\eqref{eq:overview_degree_bound}. Otherwise an agent $v\in H$ lies in $D$.
An agent $v\in V\setminus H$ lies in $B$, $C$, $D$, respectively, if $c_v\in \left(\varepsilon/(2n),\, \frac{\sigma n(\sigma n+1)}{\funcPseudoCore}\right]$, $c_v\in \left[0, \frac{\varepsilon}{2n}\right]$, $c_v\in \left(\frac{\sigma n(\sigma n+1)}{\funcPseudoCore},\,+\infty\right)$. Note that the costs of agents in $D\setminus H$ are so high that they cannot satisfy the inequality~\eqref{eq:overview_degree_bound}, due to the definition of the set $H$. 

We now \textbf{know} the sets  $A$, $B$, $C$ and $D$. Now we can obtain a lower bound~$d_v$ on $\deg_{S'}(v)$ for all $v\in B\cap S'$ using the inequality~\eqref{eq:overview_degree_bound}. So we can set a final linear program.

   \begin{align}
      \min_{\{x_v\}_{v\in V}} \quad& \sum_{v\in A} \frac{c_v}{ \deg_{S'}(v)}  \cdot x_v \label{eq:overview_final_obj}\\
    \text{subject to} \quad& \sum_{v\in H}\deg_{S'}(v)\cdot x_v \ge 2\cdot |E(S'\cap H)|\label{eq:overview_final_const1}\\
       & \sum_{u \in N(v)} x_u \geq \deg_{S'}(v) &&\text{  for all } v\in A\label{eq:overview_final_const2}\\
       & \sum_{u \in N(v)} x_u \geq d_v\cdot x_v &&\text{  for all } v\in B\label{eq:overview_final_const3}\\
       & x_v=1 &&\text{  for all } v\in C\label{eq:overview_final_const4}\\
       & x_v=0 &&\text{  for all } v\in D\label{eq:overview_final_const5}\\
        & 0 \leq x_v\leq 1 &&\text{  for all } v\in V\,.
\end{align}

\textbf{Randomized Rounding.}
We can now solve the final linear program and obtain an optimal LP solution. Then we can try to apply randomized rounding to the found optimal solution. Unfortunately, due to concentration issues, this strategy does not immediately work. To make it work, we need to strengthen the constraint~\eqref{eq:overview_final_const3} in the final linear program, and before randomized rounding we need to preprocess the optimal solution for the final linear program. 

We remark that the objective function~\eqref{eq:overview_final_obj} of the final linear program still does not account for all agents, but only for the agents in $A$. However, we can show that if the randomized rounding has sufficient concentration with respect to the constraint~\eqref{eq:overview_final_const3} then the costs of all agents outside of $A$ can be ignored.

It is challenging to achieve the sufficient concentration with respect to the constraint~\eqref{eq:overview_final_const3}. For this we need to strengthen the previously obtained bounds $d_v$, $v\in B$. We prove that $d_v$, $v\in B$ can be set to the maximum between $\Omega(n/\ln\ln(n))$ and the previous bound obtained only from~\eqref{eq:overview_degree_bound}.

After that to achieve the concentration, we show how to obtain a near optimal (and also "near feasible") solution $x^*$ to the final linear program, where each non-zero $x^*_v$, $v\in A\cup B$ is sufficiently large, i.e. non-zero variables $x^*_v$, $v\in A\cup B$ are not too close to $0$. All these properties allow us to prove the desired concentration results and so successfully round the solution $x^*$ to a near optimal solution $S''$.

Finally, we recall that we assumed we \textbf{knew} the set $H$, the degrees $\deg_{S'}(v)$, $v \in H$, and $|E(S' \cap H)|$. As remarked earlier, these values are ultimately obtained through sampling and exhaustive inspection. Thus, we repeat our process many times and show that with high probability we will eventually obtain our approximately optimal rounded solution.

\textbf{Main Challenges.}
Let us summarize the main challenges  and the way we overcome them in our approach. 
 \textbf{Challenge 1.}
 It is unclear whether there exists a near optimal set $S'$ such that $\deg_{S'}(i)$ is $\Omega(n)$ for all $i\in S'$. Generally speaking, one can ignore every "cheap'' agent $i\in S'$ because their contribution to $L(\cdot)$ is negligible regardless of $\deg_{S'}(i)$. What about "expensive'' agents? Some straightforward modifications of techniques in \cite{deo2024supermodular} show that each "expensive'' agent $i\in S'$ has to satisfy $\deg_{S'}(i)=\Omega(\sqrt n)$, but that bound is too weak for later steps to work.
 \textbf{Our approach.} We introduce a concept of \emph{pseudo-core}, see  Section~\ref{sec:near_optimal_solution}. We then develop a procedure, called \emph{iterative pseudo-coring}, that iteratively increases costs for "expensive'' agents and removes agents that are not a part of the desired pseudo-core, see Algorithm~\ref{alg:iteratedCoring}. The iterative nature of the procedure allows us to control changes in functions that we use to approximate $L(\cdot)$. The total cost increase for "expensive'' agents guarantees that remaining "expensive'' agents have sufficiently high degree $\deg_{S'}(\cdot)$. In particular, we show that each "expensive'' agent $i\in S'$ has to satisfy $\deg_{S'}(i)=\Omega(n/\ln\ln n)$.  
    
 \textbf{Challenge 2.}
 Sufficiently accurate estimates for $\deg_{S'}(i)$, $i\in V$ can be efficiently obtained only for $i$ with $\deg_{S'}(i)=\Omega(n)$. 
\textbf{Our approach.} We still use the accurate estimate for each agent $i$ with $\deg_{S'}(i)=\Omega(n)$ which can be obtained through sampling. For all "expensive'' agents in $S'$ without accurate estimates for $\deg_{S'}(i)$, our iterative pseudo-coring guarantees a sufficiently high lower bound on $\deg_{S'}(\cdot)$ and we use this lower bound instead. 

\textbf{Challenge 3.}
The absence of accurate estimates for some "expensive'' agents in $S'$ prevents us from linearizing the terms for those agents in the function $L(\cdot)$. Moreover, since there are no accurate estimates for $\deg_{S'}(i)$ for many agents $i\in V$, we need to design a new LP. This new LP makes any naive rounding approach infeasible due to concentration requirements.
\textbf{Our approach.} With concerns about linearizing $L(\cdot)$ in mind, we demonstrate that we can ignore the terms in $L(\cdot)$ that correspond to agents without accurate estimates of their degrees. To circumvent the concentration issues, we design a preprocessing procedure for an optimal LP solution. Our preprocessing procedure guarantees that "expensive'' agents have their variable set to zero or have their variable set to sufficiently large value. The LP solution achieved by the preprocessing allows us to obtain necessary concentration results.

\section{Formal Statements}
Our algorithm heavily relies on the existence of a well-structured near optimal solution $S'$. We neither know the set $S'$ nor intend to find it. However, its existence allows us to set up an LP that is \emph{up to some extent} a relaxation of our problem.
The next lemma captures the desired properties of a near optimal solution $S'$. In all remaining, parts of the paper $S'$ refers to a set satisfying the properties in Lemma~\ref{lem:structure_core}.
\begin{lemma}\label{lem:structure_core}
    There is a set $S'\subseteq V$ which has the following properties
    \begin{enumerate}[(i)]
        \item \label{item:structure_core1} $g(S') \geq OPT- \varepsilon$
        \item \label{item:structure_core2} ${\big| \{v \in S' \mid \deg_{S'}(v) \geq  \frac{\varepsilon}{6}n \} \big| \geq \frac{\varepsilon}{6}n}=\funcPseudoCoreInternal$
        \item \label{item:structure_core3} for all $v \in S'$ either 
        \begin{enumerate}
            \item \label{item:structure_core3a} 
                $\deg_{S'}(v) \geq {\frac{(1 - \varepsilon)^2 \cdot \varepsilon^5 \ln(2)}{36e} \cdot \frac{n}{\ln \ln(n)}} = \delta^{-1} \cdot \frac{n}{\ln\ln(n)}$, \quad and \\
                $\deg_{S'}(v)(\deg_{S'}(v)+1)/c_v \geq \funcPseudoCore$, \quad or
            \item \label{item:structure_core3b} $c_v \leq \varepsilon/(2 n)$ 
        \end{enumerate}
    \end{enumerate}
\end{lemma}

The LPs in the next theorem and lemma, Lemma~\ref{lem:main_supermodular1} and Theorem~\ref{thm:main_supermodular2}, will be used by our PTAS but these LPs are formulated using the set $S'$ from Lemma~\ref{lem:structure_core} above. Unfortunately, Lemma~\ref{lem:structure_core} guarantees us the existence of such a set $S'$ but does not provide an efficient way to find it. After a closer look, we will see that LPs in Lemma~\ref{lem:main_supermodular1} and Theorem~\ref{thm:main_supermodular2} rely only on the following information $\hat d_v$ for $v\in A$, $d_v$ for $v\in B$, $|E(S')|$ and the sets $A$, $B$, $C$, $D$. While $|E(S')|$ can take at most $O(n^2)$ values, i.e. $0$, $1$, \ldots, $n(n+1)/2$, and so we can try out all of these $O(n^2)$ possible values. 
We will see that the remaining  information (other than $E(S')$) can be efficiently "obtained" without knowing $S'$. This can be done through sampling or exhaustive inspection due to the properties guaranteed by Lemma~\ref{lem:structure_core}. Before doing that, let us define these parameters and sets. 

 Note that in the next definition the agents are partitioned into four parts $A$, $B$, $C$ and $D$. Let us  give some intuition behind this partition. The "cheap'' agents are in $C$, i.e. the agents satisfying~\eqref{item:structure_core3b} in Lemma~\ref{lem:structure_core}. The agents in $A$ are the agents that could potentially lie in the set mentioned in~\eqref{item:structure_core2} and at the same time could satisfy~\eqref{item:structure_core3a} in Lemma~\ref{lem:structure_core}. The agents in $B$ are the agents that cannot lie in the set mentioned in~\eqref{item:structure_core2} but could satisfy~\eqref{item:structure_core3a} in Lemma~\ref{lem:structure_core}.  The agents in $D$ can satisfy neither~\eqref{item:structure_core3a} nor~\eqref{item:structure_core3b} in Lemma~\ref{lem:structure_core}. We note that the second bound of~\eqref{item:structure_core3a} in Lemma~\ref{lem:structure_core} is equivalent to $\deg_{S'}(v) \geq \sqrt{c_i\cdot \funcPseudoCore+(1/4)}-(1/2)$.  Here, $\hat d_i$, $i\in H$  represent estimates of degrees $\deg_{S'}(i)$.

\begin{definition}\label{def:lpParam}
    Let $S' \subseteq V$ be a set as in Lemma~\ref{lem:structure_core}. Define $C=\{i\in V\,:\, c_i\leq \varepsilon/(2n)\}$. Let us consider $H$ such that $\{i\in S'\setminus C\,:\, \deg_{S'}(i)\geq \sigma\cdot n\}\subseteq H\subseteq \{i\in V \setminus C \,:\, \deg_{S'}(i)\geq (1-\varepsilon)\sigma\cdot n/(1+\varepsilon)\}$. Let $\hat d_i$, $i\in H$ be such that we have $(1-\varepsilon)\deg_{S'}(i) \leq \hat d_i\leq (1+\varepsilon)\deg_{S'}(i)$.
    Let $d_i$, $i\in V \setminus C$ be defined as
    \[
    d_i := \max\left(\frac{(1 - \varepsilon)^2 \cdot \varepsilon^5 \ln(2)}{36e} \cdot \frac{n}{\ln \ln(n)},\qquad \sqrt{c_i\cdot \funcPseudoCore+\frac{1}{4}}-\frac{1}{2}\right) \,.
    \]
    Now define $A := \{i\in H\,:\, \hat d_i / (1-\epsilon) \geq d_i\}$, $D := \left( H \setminus A \right) \cup \{ i \in V\setminus\left(H\cup C\right) \,:\, d_i \geq \sigma n \}$, and $B:=V\setminus\left(A \cup C \cup D \right)$.  
\end{definition}

Through Lemma~\ref{lem:sampler} and the discussion after it, we see that with a sufficiently high probability we can efficiently obtain the set $H$ and  $\hat d_i$, $i\in H$ through sampling.
Note that with $H$ and  $\hat d_i$, $i\in H$ at hand, we immediately obtain $\hat d_v$ for $v\in A$, $d_v$ for $v\in B$ and the sets $A$, $B$, $C$, $D$. For now, we defer the discussion about obtaining $H$ and  $\hat d_i$, $i\in H$ until Lemma~\ref{lem:sampler}.

Now Lemma~\ref{lem:main_supermodular1} establishes that the objective function in our LP reflects the value of $L(\cdot)$. Note that in the objective function we consider only agents in $A$, and we ignore those in $B$ and $C$. We will later show the agents in $B$ and $C$ do not contribute much to $L(\cdot)$.

\begin{lemma}\label{lem:main_supermodular1}
    Consider the following LP and let $OPT_{LP}$ be its optimal value.
    \begin{align}
      \min_{\{x_v\}_{v\in V}} \quad& \sum_{v\in A} \frac{c_v}{\hat d_v}  \cdot x_v \label{lp_intro:objective}\\
    \text{subject to} \quad& \sum_{v\in A}\hat{d}_v x_v \ge 2(1-\varepsilon) \cdot |E(S')| \label{lp_intro:edges}\\
       & \sum_{u \in N(v)} x_u \geq \left(\frac{1}{1+\varepsilon}\right) \cdot \hat d_v &&\text{  for all } v\in A \quad \label{lp_intro:high_degree}\\
        & \sum_{u \in N(v)} x_u \geq \left(\frac{1}{1+\varepsilon}\right) \cdot d_v \cdot x_v &&\text{  for all } v\in B \quad \label{lp_intro:low_degree} \\
        & x_v=1   &&\text{  for all } v\in C \label{lp_intro:cheap} \\
        & x_v=0   &&\text{  for all } v\in D \label{lp_intro:expensive} \\
        & 0 \leq x_v\leq 1 &&\text{  for all } v\in V \quad  \label{lp_intro:bounds}
\end{align}

 Then $OPT_{LP}$ is less than or equal to $\frac{1}{1-\varepsilon} \left(1- L(S')\right)$, i.e. less than or equal to \[\frac{1}{1-\varepsilon}\left( \sum_{v \in S'} \frac{c_v}{\deg_{S'}(v)} \right)\,.\]
\end{lemma}

Before we move ahead let us explain the intuition behind the constraints in the LP~\eqref{lp_intro:objective}-\eqref{lp_intro:bounds} from Lemma~\ref{lem:main_supermodular1}. Let us assume that we are dealing with an integral solution $\bar{x}$ and the corresponding set $\bar{S}$, $\bar{S}\subseteq V$.
The objective function~\eqref{lp_intro:objective} is estimating the portion of the principal's reward transferred to the agents in $A$. Transfers done to all other agents are not taken into account by the LP. Later, we see that with such an objective function $OPT_{LP}$ is sufficiently close to $1-L(\bar S)$. The constraint~\eqref{lp_intro:edges} guarantees that the value $R(\bar{S}\cap A)$ is at least $(1-\varepsilon)\cdot R(S')$. The constraints~\eqref{lp_intro:high_degree} and~\eqref{lp_intro:low_degree} guarantee that for every $v\in A$ and $v\in \bar{S}\cap B$ the value $\deg_{\bar{S}}(v)$ has a lower bound similar to the one in Lemma~\ref{lem:structure_core}. The constraint~\eqref{lp_intro:cheap} guarantees that all agents in $V$ with sufficiently small costs are included in $\bar{S}$. The constraint~\eqref{lp_intro:expensive} guarantees that we do not include agents that do not have a lower bound on $\deg_{\bar{S}}(v)$ similar to the one in Lemma~\ref{lem:structure_core}.

Next, Theorem~\ref{thm:main_supermodular2} establishes that there is a near-feasible solution $x^*$ which we can compute efficiently. This $x^*$ satisfies necessary properties for concentration in randomized rounding, is optimal, and is near feasible.

\begin{theorem}\label{thm:main_supermodular2}
 Let $OPT_{LP}$ be the optimal value of the LP~\eqref{lp_intro:objective}-\eqref{lp_intro:bounds}. Then
 there exists $x^*$ satisfying 
   \begin{align}
   &x^*_v>0 \implies x^*_v\geq 1/n^{1/2}&&\text{ for all } v\in A \label{lp_relax:nonzero_a} \\
    &x^*_v>0 \implies \sum_{u\in N(v) }x^*_u \geq \funcConcentration \qquad &&\text{ for all } v\in B \label{lp_relax:nonzero} \quad,
    \end{align}
and also
    \begin{align}
       &\sum_{v\in A} \frac{c_v}{\hat d_v}  \cdot x^*_v\leq  OPT_{LP}\label{lp_relax:objective}\\
    & \sum_{v\in A}\hat{d}_v x^*_v \ge 2(1-\varepsilon)^2 \cdot |E(S')| \label{lp_relax:edges}\\
       & \sum_{u \in N(v)} x^*_u \geq \left(\frac{(1-\varepsilon)^2}{1+\varepsilon}\right) \cdot \hat d_v &&\text{  for all } v\in A \quad \label{lp_relax:high_degree}\\
        & \sum_{u \in N(v)} x^*_u \geq \left(\frac{1-2\varepsilon (1+\varepsilon)}{1+\varepsilon}\right) \cdot d_v \cdot x^*_v &&\text{  for all } v\in B \quad \label{lp_relax:low_degree} \\
        & x_v^*=1   &&\text{  for all } v\in C \label{lp_relax:cheap} \\
        & x_v^*=0   &&\text{  for all } v\in D \label{lp_relax:expensive} \\
        & 0 \leq x_v^* \leq 1 &&\text{  for all } v\in V \quad  \label{lp_relax:bounds}\,.
\end{align}
Moreover, such an $x^*$ can be efficiently computed.
\end{theorem}

Finally, Theorem~\ref{thm:main_supermodular3} establishes that we can indeed round $x^*$ to a near optimal solution.

\begin{theorem}\label{thm:main_supermodular3}
 Let us be given a vector $x^*$ as in Theorem~\ref{thm:main_supermodular2}, i.e. a vector $x^*$ satisfying~\eqref{lp_relax:nonzero_a}-\eqref{lp_relax:bounds}. Let us construct a set $S''$ by independently including each $v\in A\cup B$ into the set $S''$ with probability $x^*_v$; and afterwards including each $v\in C$ into the set $S''$ if $N(v)\cap S''\neq \varnothing$. Then with probability at least $1- \sqrt{\varepsilon} - 1/n$ we obtain $g(S'')\geq g(S')-4\sqrt{\epsilon} \geq OPT - 5\sqrt{\epsilon}$.
\end{theorem}

We observe that Theorem~\ref{thm:main_supermodular3} allows us to independently  repeat the rounding in~Theorem~\ref{thm:main_supermodular3}  to achieve a good approximation with high probability. In particular, if we independently run our algorithm $\ln(1/n) / \ln(\sqrt{\varepsilon} + 1/n)$ times then we obtain a set $S''$ with $g(S'') \geq OPT - 5\sqrt{\epsilon}$ with probability at least $1 - 1/n$. This yields our desired Main Theorem~\ref{thm:intro_supermodular}.

Finally, let us show that we can efficiently set up the constraints in the LP~\eqref{lp_intro:objective}-\eqref{lp_intro:bounds}, and so we can then use Theorem~\ref{thm:main_supermodular3} and Theorem~\ref{thm:main_supermodular2} to obtain a near optimal solution.
Let us explain the way we set up the LP~\eqref{lp_intro:objective}-\eqref{lp_intro:bounds}. One of the parameters present in the LP  is $|E(S')|$. As mentioned above,  $|E(S')|$ can take only integral values between $0$ and $n(n-1)/2$ and hence permits us to do exhaustive search. It remains to verify that the set $H$ and the values $\hat{d}_v$ also can be found. For this, we use the oblivious samplers of \cite{deo2024supermodular,daskalakis2011oblivious}. 

\begin{lemma}[\cite{deo2024supermodular,daskalakis2011oblivious}]\label{lem:sampler}
    Let $K\subseteq V$ be any subset of nodes with $|K| \geq \frac{\varepsilon}{6} \cdot n$.
    There is an $O\left(n^{\poly\left(\frac{1}{\varepsilon}\right)}\right)$-time algorithm unaware of the nodes $K$ that computes $t = O\left(n^{\poly\left(\frac{1}{\varepsilon}\right)}\right)$ different marginal value estimates $\estMarginal{v}^1,\dots, \estMarginal{v}^t$ of $\marginal{v}{K}$ for each node $v \in V$ such that with probability at least $1-\frac{1}{\poly(n)}$ there exists an $j \in [t]$ whose estimates satisfy,
    \begin{enumerate}[(i)]
        \item \label{item:sampler1} $\Pr\left[ \estMarginal{v}^j < (1+\varepsilon) \sigma \cdot n \right]\geq 1 - \frac{1}{2 n^3}$, \quad $\forall v\in V$ with $\marginal{v}{K} < \frac{\varepsilon (1+\varepsilon)\sigma}{9} n$
        
        \item \label{item:sampler2}  $\Pr\left[  (1-\varepsilon) \cdot \marginal{v}{K} \leq \estMarginal{v}^j \leq  (1+\varepsilon) \cdot \marginal{v}{K} \right]\geq 1 - \frac{1}{2  n^3}$, \quad $\forall v\in V$ with $\marginal{v}{K} \geq \frac{\varepsilon (1+\varepsilon) \sigma}{9} n$
    \end{enumerate}
\end{lemma}

We remark that Lemma~\ref{lem:sampler} allows us to set the LP~\eqref{lp_intro:objective}-\eqref{lp_intro:bounds} even without knowing the set $S'$. Note that by Lemma~\ref{lem:structure_core}~\eqref{item:structure_core2} we have $|S'|\geq \varepsilon n /6$, and so we can apply Lemma~\ref{lem:sampler}  for $K$ equal to $S'$. Then we run the algorithm outlined in Theorem~\ref{thm:main_supermodular2} and Theorem~\ref{thm:main_supermodular3} independently for every $j\in [t]$ by setting $\hat{d}_v:=\estMarginal{v}^j$ and then selecting  $H:=\{ i\in V \setminus C\,:\, \hat{d}_i\geq \sigma n (1-\varepsilon)\}$. Let us argue that with high probability at least for one $j\in [t]$ the selected above $\hat{d}_i$, $i\in H$ and the set $H$ are correct. Lemma~\ref{lem:sampler} guarantees that with probability $1-1/\poly(n)$ there exists $j\in [t]$ such that both~\eqref{item:sampler1} and~\eqref{item:sampler2} are satisfied. Conditioned on this, by~\eqref{item:sampler1} and~\eqref{item:sampler2} with probability at least $1-1/(2n^2)$ we have \[\{i\in S'\setminus C\,:\, \deg_{S'}(i)\geq \sigma\cdot n\}\subseteq H\subseteq \{i\in V\setminus C\,:\, \deg_{S'}(i)\geq (1-\varepsilon)\sigma\cdot n/(1+\varepsilon)\}\]
and for all $i\in H$ we have $(1-\varepsilon)\deg_{S'}(i) \leq \hat d_i\leq (1+\varepsilon)\deg_{S'}(i)$.

\subsubsection*{Acknowledgements.}
Research of Kanstantsin Pashkovich
was supported in part by  Discovery Grants Program RGPIN-2020-04346, Natural Sciences and Engineering Research
Council (NSERC) of Canada.
\newpage

\bibliographystyle{splncs04}
\bibliography{references}

\newpage

\appendix

\section{Appendix: Omitted Proofs}

\subsection{Linear Relaxation and Rounding}\label{sec:lp_rounding}

In this section, we assume that we have the information about $S'$ that is required to set up the LP~\eqref{lp_intro:objective}-\eqref{lp_intro:bounds}. With this assumption, we show how to find a near optimal solution. To do this we prove  Lemma~\ref{lem:main_supermodular1}, Theorem~\ref{thm:main_supermodular2} and Theorem~\ref{thm:main_supermodular3}.

\subsubsection{Proof of Lemma~\eqref{lem:main_supermodular1}}

    Let $S'$ be a set of agents satisfying Lemma~\ref{lem:structure_core}. Clearly, the characteristic vector $x'$ of the set $S'\cup C$ is a feasible solution for the LP~\eqref{lp_intro:objective}-\eqref{lp_intro:bounds}. For comparison's sake, we note that the bound of $\deg_{S'}(v)(\deg_{S'}(v)+1) \geq \funcPseudoCore \cdot c_v$ implies $\deg_{S'}(v)^2/c_v \geq \frac{(1-\varepsilon)\varepsilon^4}{36} \cdot n^2$. Moreover, the value of $x'$ equals 
\begin{equation}
\label{eq:desiredSetFeasibility}
    \sum_{v\in A} \frac{c_v}{\hat d_v}  \cdot x'_v\leq  \sum_{v\in A}\frac{c_v}{(1-\varepsilon)\deg_{S'}(v)}\leq \sum_{v\in S'}\frac{c_v}{(1-\varepsilon)\deg_{S'}(v)}=\frac{1}{1-\varepsilon} \left(1- L(S')\right)\,,
\end{equation}
showing the desired statement.

\subsubsection{Proof of Theorem~\ref{thm:main_supermodular2} }

In this section, we design an efficient preprocessing procedure, i.e. Algorithm~\ref{alg:fractionalCoring},  for feasible solutions of  the LP~\eqref{lp_intro:objective}-\eqref{lp_intro:bounds}. We apply this procedure to an optimal solution of the LP~\eqref{lp_intro:objective}-\eqref{lp_intro:bounds}. The following lemma directly implies Theorem~\ref{thm:main_supermodular2}. 

Given a vector $x$, Algorithm~\ref{alg:fractionalCoring} iteratively sets to zero entries $x_v$ which correspond to the inequalities in the LP~\eqref{lp_intro:objective}-\eqref{lp_intro:bounds} without sufficient concentration guarantee. Intuitively, in the proof of  Lemma~\ref{lem:fractionalCoring} we show that Algorithm~\ref{alg:fractionalCoring} does not set too many coordinates to zero. In particular, Algorithm~\ref{alg:fractionalCoring} consists of two loops. The first loops "zeroes out" in $x$ coordinates that sum up to at most $n^{1/2}$. Later, we show that $\funcConcentrationPlus  \ln(n)$ is a rough upper bound on the sum of coordinates "zeroed" in the second loop. These two upper bounds allow us to obtained the desired statement.

\begin{algorithm}[h!]
    \caption{Fractional Coring}\label{alg:fractionalCoring}
    \SetKwFunction{FractionalCoring}{FractionalCoring}
    \Input{Feasible solution $x$ to LP~\eqref{lp_intro:objective}-\eqref{lp_intro:bounds}}
    $x^* \gets x,
     j \gets 1$ \;
     \For{$v \in A$}{\label{line:removingA}
         \uIf{$x_v^* \leq 1/n^{1/2}$}{
            $x_v^* \gets 0$ \;
            }
         }
    \While{true}{\label{line:fractional_coring_while_loop}
        $\displaystyle v \gets \argmin_{\substack{v \in B \\ x^*_v > 0}} \left( \sum_{u \in N(v)} x^*_u \right)$  \;
        $\displaystyle w \gets \argmin_{\substack{w \in B \\ x^*_w > 0}} \left( \varepsilon d_w  x^*_w - \sum_{i=1}^{j-1} x_{v_i} \right)$ \;
        \uIf{$\displaystyle \left( \sum_{u \in N(v)} x^*_u \right) < \funcConcentration$} {
            $x^*_v \gets 0,
             v_j \gets v,
             j \gets j+1$ \;
        }
        \uElseIf{$\displaystyle \left( \varepsilon d_w  x^*_w - \sum_{i=1}^{j-1} x_{v_i} -n^{1/2}\right) < 0$\label{line:fractionalcoring}} {
                $x^*_w \gets 0,
                 v_j \gets w,
                 j \gets j+1$ \; 
            }
        \uElse{
            Exit loop \;
        }
    }  
    \Output{$x^*$}
\end{algorithm}
\begin{lemma}\label{lem:fractionalCoring}
    Given an optimal solution $x$ for the LP~\eqref{lp_intro:objective}-\eqref{lp_intro:bounds}, Algorithm~\ref{alg:fractionalCoring} outputs a vector $x^*$ such that~\eqref{lp_relax:nonzero_a}-\eqref{lp_relax:bounds} hold.
\end{lemma}

\begin{proof}

We first consider the changes caused by the loop on line~\ref{line:removingA}. Let $X$ be the set of $v$ such that $x_v^*$ was set to 0 in this loop. Observe that we have $\sum_{v \in X} x_v \leq n^{1/2}$. This implies
\[ \sum_{u \in N(v)} x_u \leq \sum_{u \in N(v)} x^*_u + n^{1/2}\qquad \text{ for all } v\in A \cup B \quad  \]
at the time point when the algorithm moves to line~\ref{line:fractional_coring_while_loop}. Thus, we obtain
\begin{equation*}
    \sum_{u \in N(v)} x^*_u \geq \left(\frac{1-\varepsilon}{1+\varepsilon}\right) \cdot \hat d_v \qquad \text{ for all } v\in A \quad
\end{equation*}
at this point of the algorithm. Similarly, notice that at this point we have
\[
\sum_{v \in A} \hat d_v x_v - \sum_{v \in A} \hat d_v x_v^* \leq n^{3/2} \,,
\]
which implies that~\eqref{lp_relax:edges} is true at this point. Now we consider how the rest of the algorithm changes $x^*$ from this point onwards, i.e. from the moment when the algorithm moves to line~\ref{line:fractional_coring_while_loop} for the first time.

Let $J$ be the  value of $j$ immediately before Algorithm~\ref{alg:fractionalCoring} terminated. 
    Observe that at the end of the algorithm we have $\sum_{u \in N(v)} x^*_u \geq \funcConcentration$ for all $v \in B$, $x^*_v > 0$, implying that~\eqref{lp_relax:nonzero} holds for $x^*$. Note,  the entries of $x^*$ indexed by $v\in A\cup C$  did not change after the first loop of the algorithm, implying that~\eqref{lp_relax:nonzero_a},~\eqref{lp_relax:objective},~\eqref{lp_relax:edges} and~\eqref{lp_relax:cheap} hold for $x^*$, and we have $x^*_v\geq 1/n^{1/2}$ or $x^*_v=0$ for all $v\in A$. By line~\ref{line:fractionalcoring} in Algorithm~\ref{alg:fractionalCoring}, for all $v\in B$ with $x^*_v>0$ we have
    \begin{equation*}
        \varepsilon \cdot d_v \cdot x^*_v > \sum_{i=1}^J x_{v_i}+n^{1/2}\,.
    \end{equation*}
    So $x^*$ satisfies also~\eqref{lp_relax:low_degree}. Clearly, the output $x^*$  satisfies~\eqref{lp_relax:expensive} and~\eqref{lp_relax:bounds}. In the remaining part of the proof we show that $x^*$ satisfies also~\eqref{lp_relax:high_degree}.

 Consider the $j$th iteration, i.e. the iteration right before $x^*_{v_j}$ gets set to 0. For compactness in this proof, we introduce the following notation
 \begin{gather*}
     \bar{\kappa}:= \funcConcentration, \quad
      \bar{\delta} := \funcInverseDegreeBound, \quad
     \bar{\delta_\varepsilon} :=  \bar{\delta} / \varepsilon \,.
 \end{gather*}
 We have the following facts:
\begin{enumerate}[(i)]
    \item \label{fractionalcoring_item1} $\displaystyle \sum_{u \in N(v_j)} x^*_u < \bar{\kappa}$, \quad or \quad
    $\displaystyle x^*_{v_j} / \bar{\delta_\varepsilon} \leq \varepsilon d_{v_j} x^*_{v_j} \leq \sum_{i=1}^{j-1}  x_{v_i} + n^{1/2}$
    \item \label{fractionalcoring_item2}  $\displaystyle \sum_{u \in N(v_j)} x_u \leq \sum_{u \in N(v_j)} x^*_u + \sum_{i=1}^{j-1} x_{v_i}  + n^{1/2}$ 
    \item \label{fractionalcoring_item3}  $\displaystyle x_{v_j} / \bar{\delta} \leq \sum_{u \in N(v_j)} x_u$.
\end{enumerate}
Here,~\eqref{fractionalcoring_item1} follows from the fact that $x^*_{v_j}$ is set to $0$ in the current iteration. Due to the definition of $v_1$, \ldots, $v_{j-1}$  we have that $ \sum_{u \in V} x_{u} - \sum_{u \in V} x^*_u $ is at most $\sum_{i=1}^{j-1}  x_{v_i} + n^{1/2}$, resulting in~\eqref{fractionalcoring_item2}. The inequality~\eqref{fractionalcoring_item3} holds since $x^*$ after the loop on line~\ref{line:removingA} is feasible for the LP~\eqref{lp_relax:objective}-\eqref{lp_relax:bounds}. 

We combine~\eqref{fractionalcoring_item1}, \eqref{fractionalcoring_item2} and \eqref{fractionalcoring_item3} to obtain that one of the following inequalities holds
\[
    x^*_{v_j} \leq \bar{\delta} \cdot\left(\bar{\kappa} + n^{1/2} + \sum_{i=1}^{j-1} x_{v_i} \right)
    \qquad\text{ or }\qquad
     x^*_{v_j} \leq  \bar{\delta_\varepsilon} \cdot \left(\sum_{i=1}^{j-1}  x_{v_i} + n^{1/2}\right)\,,
\]
depending on the outcome in~\eqref{fractionalcoring_item1}.

Let us prove that  
\begin{equation}\label{eq:fractionalcoring}
     x_{v_j} \leq \left(\bar{\kappa} + n^{1/2}\right) \bar{\delta_\varepsilon} \left( \bar{\delta_\varepsilon}+1 \right)^{j-1}
\end{equation}
holds for every $j=1,\ldots, J-1$. We prove the inequality~\eqref{eq:fractionalcoring} by induction. It is straightforward to notice that 
$ x_{v_1} \leq \bar{\delta_\varepsilon} \cdot (\bar{\kappa} + n^{1/2})$. Let us assume that for all $1 \leq i \leq j$ the inequality~\eqref{eq:fractionalcoring} holds. 

\textbf{Case $\sum_{u \in N(v_{j+1})} x^*_u < \bar{\kappa}$.} In this case, we have
\begin{align*}
    & x^*_{v_{j+1}} / \bar{\delta}
    \leq \bar{\kappa} + n^{1/2} + \sum_{i=1}^{j} x_{v_i}
    \leq {\left(\bar{\kappa} + n^{1/2}\right)} + \sum_{i=1}^{j}{\left(\bar{\kappa} + n^{1/2}\right)} \bar{\delta_\varepsilon} \left( \bar{\delta_\varepsilon}+1 \right)^{i-1} \\
    =& {\left(\bar{\kappa} + n^{1/2}\right)} + {\left(\bar{\kappa} + n^{1/2}\right)} \cdot  \bar{\delta_\varepsilon} \cdot \sum_{i=1}^{j} \left(\bar{\delta_\varepsilon}+1\right)^{i-1} \\
    =& {\left(\bar{\kappa} + n^{1/2}\right)} \left( 1 + \bar{\delta_\varepsilon} \cdot \frac{(\bar{\delta_\varepsilon}+1)^j - 1} {\bar{\delta_\varepsilon}} \right)
    = {\left(\bar{\kappa} + n^{1/2}\right)} \cdot \left( \bar{\delta_\varepsilon}+1 \right)^j\,,
\end{align*}
leading us to the inequality~\eqref{eq:fractionalcoring} by rearranging for $x^*_{v_{j+1}}$.  

\textbf{Case $x^*_{v_{j+1}} / \bar{\delta_\varepsilon} \leq \sum_{i=1}^j { x_{v_i} + n^{1/2}}$.} In this case, we have
\begin{align*}
    &x^*_{v_{j+1}}
    \leq \bar{\delta_\varepsilon}  \cdot {\left(\sum_{i=1}^{j}  x_{v_i} + n^{1/2}\right)}
    \leq  \bar{\delta_\varepsilon} 
    \cdot \left(\sum_{i=1}^{j} {\left(\bar{\kappa} + n^{1/2}\right)} \bar{\delta_\varepsilon}\left(\bar{\delta_\varepsilon}+1\right)^{i-1} + n^{1/2} \right) \\
    =&\, \bar{\delta_\varepsilon} \cdot \left({ {\left(\bar{\kappa} + n^{1/2}\right)}\left(\left( \bar{\delta_\varepsilon}+1\right)^j - 1 \right) + n^{1/2} }\right)
    \leq {\left(\bar{\kappa} + n^{1/2}\right)} \bar{\delta_\varepsilon} \left( \bar{\delta_\varepsilon}+1\right)^j\,,
\end{align*}
leading us again to the inequality~\eqref{eq:fractionalcoring}. Thus, the inequality~\eqref{eq:fractionalcoring} always holds.

By~\eqref{eq:fractionalcoring}, we have
\begin{align*}
    &\sum_{j=1}^J x_{v_j}\leq \sum_{j=1}^J {\left(\bar{\kappa} + n^{1/2}\right)}\bar{\delta_\varepsilon}\left(\bar{\delta_\varepsilon}+1\right)^{j-1}
    = {\left(\bar{\kappa} + n^{1/2}\right)} \left(\left(\bar{\delta_\varepsilon}+1\right)^J - 1\right) \\
    &\leq {\left(\bar{\kappa} + n^{1/2}\right)} \left(\bar{\delta_\varepsilon}+1\right)^J
    = {\left(\bar{\kappa} + n^{1/2}\right)} \left( \frac{36e}{(1 - \varepsilon)^2 \cdot \varepsilon^6 \ln(2)} \cdot \frac{\ln \ln(n)}{n} + 1 \right)^J \\
    &\leq {\left(\bar{\kappa} + n^{1/2}\right)} \cdot \exp\left( \frac{36e}{(1 - \varepsilon)^2 \cdot \varepsilon^6\ln(2)} \cdot\ln\ln\left(n\right) \right)
    = {\left(\bar{\kappa} + n^{1/2}\right)} \cdot \ln(n)^\alpha\,,
\end{align*}
where $\alpha := \left({36e}\right) / \left({(1 - \varepsilon)^2 \cdot \varepsilon^6\ln(2)}\right)$ and so $\alpha$ is an expression depending only on $\varepsilon$. 
Here, the first inequality is due to~\eqref{eq:fractionalcoring}. The third inequality is due to the definition of $e$. Thus, we have that for the output $x^*$ the value $ \sum_{u \in V} x_{u} - \sum_{u \in V} x^*_u $ is at most ${\left(\bar{\kappa} + n^{1/2}\right)} \cdot \ln(n)^\alpha$. Note that by definition for all $i\in A$ we have $\hat{d}_i\geq (1-\varepsilon)\sigma \cdot n$. Hence, $ \sum_{u \in V} x_{u} - \sum_{u \in V} x^*_u \leq {\left(\bar{\kappa} + n^{1/2}\right)} \cdot \ln(n)^\alpha$  imply that $x^*$ satisfies also~\eqref{lp_relax:high_degree}.
\end{proof}

\subsubsection{Proof of Theorem~\ref{thm:main_supermodular3} }

Let us assume that vector $x^*$ satisfies~\eqref{lp_relax:nonzero_a}-\eqref{lp_relax:bounds}. Let us construct a set $S''$ by independently including each $v\in A\cup B$ into the set $S''$ with probability $x^*_v$. Afterwards, we include each $v\in C$ into the set $S''$ if $N(v)\cap S''\neq \varnothing$. This is precisely the randomized construction of $S''$ described in Theorem~\ref{thm:main_supermodular3}.

\begin{lemma}\label{lem:good_events}
    Let $S''$ be constructed as in Theorem~\ref{thm:main_supermodular3}. Then we have \begin{equation}\label{item:good_events_edges}
    \Pr\left[\sum_{v\in A\cap S''} \hat d_v  \geq (1-\varepsilon) \sum_{v \in A} \hat d_v x^*_v\right]\geq 1-\frac{1}{n^2}
    \end{equation}
    \begin{equation}\label{item:good_events_costs}
    \Pr\left[\sum_{v \in S'' \cap A} \frac{c_v}{\hat d_v} \leq \sum_{v \in A} \frac{c_v}{\hat d_v}x^*_v + \varepsilon\right]\geq {1-\frac{1}{n^2}}
    \end{equation}
    and for all $v\in A\cup B$ we have 
    \begin{equation}\label{item:good_events_degree}
        \Pr\left[\deg_{S''}(v) \geq (1-\varepsilon) \sum_{u \in N(v)} x^*_u\,\mid\, v\in S''\right]\geq {1-\frac{1}{n^2}}\,.
    \end{equation}
    
\end{lemma}
\begin{proof}
    Let us use Chernoff's inequality to bound the probability that the event in~\eqref{item:good_events_edges} does not occur, and obtain
    \begin{align*}
        &\Pr\left[\sum_{v\in A\cap S''} \hat d_v  < (1-\varepsilon) \sum_{v \in A} \hat d_v x^*_v\right]
        = \Pr\left[\sum_{v\in A\cap S''}\hat{d}_v < (1-\varepsilon)\EE \left[\sum_{v\in A\cap S''}\hat{d}_v \right]\right] \\
        \leq&\exp\left( -\frac{2\varepsilon^2 \EE \left[ \sum_{v\in A\cap S''}\hat{d}_v \right]^2}{n\left(\max_{v\in A} \hat{d}_v \right)^2} \right)\leq_{\star} \exp\left(-\frac{2\varepsilon^2(1-\varepsilon)^4\left( \varepsilon/6 \cdot n\right)^4}{n^3}\right)\leq  \frac{1}{n^2}\,\\
    \end{align*}
    here the inequality $\star$ follows from $\hat{d}_v\leq n$ for all $v\in A$; and from the definition of $\hat{d}_v$ in Definition~\ref{def:lpParam} and the fact that the set $S'$ satisfies~\eqref{item:structure_core2} in Lemma~\ref{lem:structure_core}.

Now we  use Hoeffding's inequality to bound the probability the event in (\ref{item:good_events_costs}) does not occur. To use this inequality we first obtain some preliminary bounds. Recall by definition for all $v \in A$ we have $(\hat d_v/(1-\varepsilon)) \geq d_v$. Then if $c_v \geq n^{1/4}$, we can expand and rearrange the definition of $d_v$ to obtain 
\[
\frac{(\hat d_v)^2}{c_v} \geq (1-\varepsilon)^2 \cdot \left( \funcPseudoCore - \frac{\hat d_v}{(c_v(1-\varepsilon))}\right) \geq (1-\varepsilon)^3 \funcPseudoCore \,,
\] and otherwise $\hat d_v / c_v \geq (1-\varepsilon)\sigma n^{3/4}$. Finally we recall $x^*_v \geq 1/n^{1/2}$ for $v \in A, x^*_v > 0$, and observe
    \begin{equation*}
        \sum_{\substack{v \in A}} \frac{c_v}{n^{3/2}} \leq \sum_{v \in A} \frac{c_v}{\hat d_v} \cdot x_v^* \leq_{\star} \sum_{v \in V} \frac{c_v}{\hat d_v} x' < 1 \,,
    \end{equation*}
    where $x'$ corresponds to the set of agents satisfying Lemma~\ref{lem:structure_core}, and inequality $\star$ follows from the feasibility of the desired set, equation~\eqref{eq:desiredSetFeasibility}. With this, we obtain
    \begin{align*}
        &\Pr\left[\sum_{v \in S'' \cap A} \frac{c_v}{\hat d_v} > \sum_{v \in A} \frac{c_v}{\hat d_v}x^*_v + \varepsilon\right]
        = \Pr\left[\sum_{v \in S'' \cap A} \frac{c_v}{\hat d_v} > \EE\left[ \sum_{v \in S'' \cap A} \frac{c_v}{\hat d_v} \right] + \varepsilon\right] \\
        \leq&\exp\left( -\frac{2\varepsilon^2}{\sum_{v \in S'' \cap A} \left( \frac{c_v}{\hat d_v} \right)^2} \right) \\
        \leq&_{\star} \exp\left( -\frac{2\varepsilon^2}{\sum_{\substack{v \in S'' \cap A \\ c_v \leq n^{1/4}}} \left( \frac{1}{(1-\varepsilon)\sigma n^{3/4}} \right)^2 + \sum_{\substack{v \in S'' \cap A \\ c_v > n^{1/4}}} \left( \frac{c_v}{(1-\varepsilon)^3\funcPseudoCore} \right)} \right)\\
        \leq&_{\star\star} \exp\left( -\frac{2\varepsilon^2}{\frac{1}{(1-\varepsilon)^2\sigma^2 n^{1/2}} + \frac{n^{3/2}}{(1-\varepsilon)^3\funcPseudoCore} } \right)
        \leq \frac{1}{n^2}
    \end{align*}
    here the inequality $\star$ follows from the above discussion bounding costs, and the inequality $\star\star$ follows from noting $|A| \leq n$ and the fact that $\sum_{v \in S'' \cap A} c_v < n^{3/2}$. Also note $\funcPseudoCore$ is at least $\varepsilon^7 n^2$.

    Now we will again use Chernoff bounds for our final concentration result. This time we bound the probability event (\ref{item:good_events_degree}) does not occur, and obtain
    \begin{align*}
        &\Pr\left[\deg_{S''}(v) < (1-\varepsilon) \sum_{u \in N(v)} x^*_u\,\mid\, v\in S''\right] \\
        &\leq \exp\left( -\frac{\varepsilon^2 \EE \left[\deg_{S''}(v)\,\mid\, v\in S''\right]}{2} \right)
        \leq_{\star} \exp\left(-\frac{\varepsilon^2\funcConcentration}{2}\right)\leq  \frac{1}{n^2}\,\\
    \end{align*}
    where inequality $\star$ is from our guaranteed bound $\EE \left[\deg_{S''}(v) \,\mid\, v\in S''\right] \geq \funcConcentration$ for any $v$ with $x_v > 0$.
\end{proof}

\subsubsection{Bound on $R(S'')$.}
Let us use Lemma~\ref{lem:good_events} to prove Theorem~\ref{thm:main_supermodular3}.
First let us estimate $R(S'')$. Let us consider the case when all events described in~\eqref{item:good_events_edges} and~\eqref{item:good_events_degree} happen. Thus,  with probability at least ${1-2/n}$ we have
\begin{align*}
    &R(S'')\cdot\maxValue=|E(S'')| \geq \frac{1}{2} \sum_{v \in S'' \cap A} \deg_{S''}(v)
    \geq_{\star}\frac{1}{2} \cdot\frac{(1-\varepsilon)^2}{1+\varepsilon} \sum_{v \in S''\cap A}\hat{d}_v \\
    &\geq_{\star\star} \frac{1}{2} \cdot\frac{(1-\varepsilon)^3}{1+\varepsilon} \sum_{v \in  A}\hat{d}_v x^*_v 
    \geq \frac{ (1-\varepsilon)^5}{(1+\varepsilon)} |E( S')|=\frac{ (1-\varepsilon)^5}{(1+\varepsilon)} \cdot R(S')\cdot\maxValue\,,
\end{align*}
where the inequality $\star$ is due to the events in~\eqref{item:good_events_degree} and since $x^*$ satisfies~\eqref{lp_relax:high_degree}, the inequality $\star\star$ holds due to the event in~\eqref{item:good_events_edges}. The last inequality holds since $x^*$ satisfies~\eqref{lp_relax:edges}.

\subsubsection{Bound on $L(S'')$.}

Let us now estimate $L(S'')$. Recall that we have 
\begin{equation}\label{eq:left_side}
  L(S'')=  1 - \sum_{v \in S''\cap A} \frac{c_v}{\deg_{S''}(v)}- \sum_{v \in S''\cap B} \frac{c_v}{\deg_{S''}(v)}-  \sum_{v \in S''\cap C} \frac{c_v}{\deg_{S''}(v)}\,.
\end{equation}

By the definition of $C$ and by the construction of $S''$, we have that $\deg_{S''}(v)\geq 1$ and $c_v\leq \varepsilon/(2n)$ for every $v\in S''\cap C$. Thus, the last term in~\eqref{eq:left_side} is always at most $\varepsilon/2$, i.e.
\[\sum_{v \in S''\cap C} \frac{c_v}{\deg_{S''}(v)}\leq \varepsilon/2\,.\]

Let us consider the case when all events described in~\eqref{item:good_events_edges}, \eqref{item:good_events_costs} and~\eqref{item:good_events_degree} happen. With probability at least ${1-2/n}$ we have that the second term in~\eqref{eq:left_side} is at most $(OPT_{LP}+\varepsilon)(1+\varepsilon)/(1-\varepsilon)^3$, i.e.
\[
\sum_{v \in S''\cap A} \frac{c_v}{\deg_{S''}(v)}\leq \frac{(OPT_{LP}+\varepsilon)(1+\varepsilon)}{{ (1-\varepsilon)^3}}\,.
\]
Indeed, 
\begin{align*}
    \sum_{v \in S''\cap A} \frac{c_v}{\deg_{S''}(v)}
    &\leq_{\star} \sum_{v \in S''\cap A} \frac{c_v}{(1-\varepsilon)\sum_{u \in N(v)} x^*_u }\leq_{\star\star} (1+\varepsilon)\sum_{v \in S''\cap A} \frac{c_v}{{ (1-\varepsilon)^3}\hat{d}_v }\\
    &\leq_{\star\star\star} \frac{1+\varepsilon}{{ (1-\varepsilon)^3}}\left(\sum_{v \in A} \frac{c_v}{\hat d_v}x^*_v + \varepsilon\right)
    \leq \frac{1+\varepsilon}{{ (1-\varepsilon)^3}}\left(OPT_{LP}+\varepsilon\right)\,,
\end{align*}
where the inequality $\star$ is due to the events in~\eqref{item:good_events_degree}, the inequality $\star\star$ is due to~\eqref{lp_relax:high_degree} for $x^*$, the inequality $\star\star\star$ is due to the event in~\eqref{item:good_events_costs}.

 Now it remains to estimate the third term in~\eqref{eq:left_side}, i.e. the term corresponding to the set $B$. As in the previous case, we assume that all events in~\eqref{item:good_events_edges}, \eqref{item:good_events_costs} and~\eqref{item:good_events_degree} happen.  Let us start with an estimate for the expected value of the term corresponding to the set $B$.
 \begin{align*}
    &\EE\left[ \sum_{v \in S'' \cap B} \frac{c_v}{\deg_{S''}(v)} \right]
    = \sum_{v \in B}\EE\left[ \frac{c_v }{\deg_{S''}(v)}\,\mid\, v\in S''\right]\cdot \Pr[v\in S''] \\
    \leq&_{\star} \sum_{v \in B}\left( \frac{(1+\varepsilon)}{(1-2\varepsilon  (1+\varepsilon))}  \frac{c_v }{d_v\cdot x^*_v}\right) \cdot x^*_v
    = \frac{(1+\varepsilon)}{(1-2\varepsilon (1+\varepsilon))} \sum_{v \in B} \frac{c_v }{d_v} \\
    \leq& \frac{(1+\varepsilon)}{(1-2\varepsilon (1+\varepsilon))} \sum_{v \in B} \frac{d_v+1}{\funcPseudoCore} \\
    \leq& \frac{(1+\varepsilon)^2}{(1-2\varepsilon (1+\varepsilon))} \sum_{v \in B} \frac{\sigma n}{\funcPseudoCore} 
    \leq \varepsilon\,,
\end{align*}
where the  inequality $\star$ is due to~\eqref{lp_relax:low_degree} for $x^*$ and due to the event in~\eqref{item:good_events_degree}. The last two inequalities are due to $c_v\leq d_v(d_v+1)/(\funcPseudoCore)$ for all $v\in B$ and due to the definition of $\sigma, \gamma$. 

Using Markov's inequality, we get
\begin{align*}
    \Pr\left[\sum_{v \in S'' \cap B} \frac{c_v}{\deg_{S''}(v)}\geq \sqrt{\varepsilon}\right]\leq \sqrt{\varepsilon}\,.
\end{align*}

Hence, with probability { at least $1-\sqrt{\varepsilon} - 1/n$} for sufficiently small $\varepsilon$ we have
\begin{align*}
&L(S'')\geq 1-\frac{1+\varepsilon}{{ (1-\varepsilon)^3}}\left(OPT_{LP}+\varepsilon\right) -\sqrt{\varepsilon}-\varepsilon/2 \\
=& 1-\frac{1+\varepsilon}{{ (1-\varepsilon)^3}}\left(\frac{1}{1-\varepsilon} \left(1- L(S')\right)+\varepsilon\right) -2\sqrt{\varepsilon}
\geq L(S')-3\sqrt{\varepsilon}\,.
\end{align*}

\subsubsection{Bound on $g(S'')$.}
Thus, with probability { at least $1- \sqrt{\varepsilon} - 1/n$} we have
\begin{align*}
g(S'')
&=L(S'')\cdot R(S'')
\geq \left(L(S')-3\sqrt{\varepsilon}\right)\cdot\frac{{ (1-\varepsilon)^5}}{1+\varepsilon} R(S') \\
&\geq L(S')R(S')-4\sqrt{\varepsilon}=g(S')-4\sqrt{\varepsilon},\
\end{align*}
using that $L(S')$ and $R(S')$ are at most $1$. Thus we 
obtain the desired statement of Theorem~\ref{thm:main_supermodular3}.

\subsection{Structured Near Optimal Solutions}\label{sec:near_optimal_solution}

In this section we prove Lemma~\ref{lem:structure_core}. The desired set $S'$ is the output of Algorithm~\ref{alg:iteratedCoring} with input given by the optimal set $S^*$ and the original costs $c_i$, $i\in V$.

Before we proceed, let us make some observations about the structure of an optimal set $S^*$. The next lemma can be shown by starting with the set $S^*$  and iteratively removing all agents with at most $\varepsilon n/3$ adjacent agents in the set.
\begin{lemma}[\cite{deo2024supermodular}]\label{lem:optimalSize}
If a set $S$ satisfies $R(S) \geq \varepsilon$, then there exists $\widetilde{S}$, $\widetilde{S}\subseteq S$  such that $R(\widetilde S)\geq R(S^*)-2\varepsilon/3$, $|\widetilde{S}|\geq \frac{\varepsilon}{3} n$ and $\marginal{i}{\widetilde{S}} \geq \frac{\varepsilon}{3}n$ for all $i \in \widetilde{S}$.
\end{lemma}

At several places, we need to bound the total cost of agents in the optimal set $S^*$ or its subsets. For this we use the following observation.

\begin{observation}\label{obs:averageCost}
If a set $S^*$ satisfies $g(S^*) \geq \varepsilon$, or just $L(S^*) > 0$, then for every $S \subseteq S^*$ we have $\sum_{i \in S} c_i \leq \maxMarginal$.
\end{observation}
\begin{proof}
    If $g(S^*) \geq \varepsilon$ then $L(S^*) > 0$. In its turn rearranging $L(S^*) > 0$ gives us
    \begin{align*}
        1 > \sum_{i \in S^*} \frac{c_i}{\marginal{i}{S^*}} \geq \sum_{\substack{i \in S}} \frac{c_i}{\maxMarginal}\,.
    \end{align*} 
\end{proof}

Let us consider a set $S\subseteq V$ and introduce its pseudo-core. Note, that the definition of the pseudo-core does not impose any conditions on the vertices in $S\cap C$, i.e. it does not impose any conditions on the vertices with sufficiently small costs.
Let us define the \emph{$\beta$-pseudo-core} of $S$ with respect to the costs $c_i'$, $i \in S$ as an inclusion-wise maximal subset $\bar{S}$ of $S$ such that
\[
 \pseudoMarginal{i}{\bar{S}}{c'} \geq \beta \qquad \text{ for all } i \in \bar{S}\setminus C\,,
\] 
alternatively one can reformulate the above condition as
\[
\deg_{\bar S}(i)\geq\sqrt{c'_i\cdot \beta+\frac{1}{4}}-\frac{1}{2}\qquad \text{ for all } i \in \bar{S}\setminus C\,.
\]

It is straightforward to check the following. It is fully analogous to similar results on core.
\begin{observation}\label{lem:coreWithCosts}
    The $\beta$-pseudo-core with costs $\{c_i'\}_{i \in S}$ of a set $S$ is in fact unique, and it can be found with the greedy Algorithm~\ref{alg:pseudoCore}.
\end{observation}

\begin{algorithm}[h!]
    \caption{PseudoCore}\label{alg:pseudoCore}
    \SetKwFunction{PseudoCore}{PseudoCore}
    \Input{Initial set $S$, costs $c_i'$ for $i \in S$, pseudo-core threshold value $\beta$}

    $\bar{S} \gets S$ \;
    \While{$\bar{S}\setminus C\neq \varnothing$}{
        $\displaystyle i \gets \argmin_{v \in \bar{S}\setminus C} \pseudoMarginal{i}{\bar{S}}{c'}$ \;
        \If{$\displaystyle \pseudoMarginal{i}{S'}{c'} < \beta$} {
            $\bar{S} \gets \bar{S} \setminus \{v\}$ \;  
        }
        \Else
        \Break
    }
    \Output{$\bar{S}$}
\end{algorithm}

Let us now explain how to construct the desired set $S'$ from an optimal set $S^*$.

First note in Algorithm~\ref{alg:iteratedCoring} we will be taking $\beta$-pseudo-cores with $\beta = \gamma(n^2 + (6/\varepsilon)n)$. Let us define $ M:=\log_2 \ln(\gamma(n+6/\varepsilon))$ and for each $k$, $1\leq k\leq M$ let us define
    \[
    \Delta_k: = \frac{\varepsilon^{\left(2 - \frac{1}{2^k}\right)} \cdot \left(\funcPseudoCore\right)^{\left(1 - \frac{1}{2^k}\right)}}  {M^{\left(2 - \frac{1}{2^{k-1}}\right)} \cdot n^{\left(2 - \frac{1}{2^k}\right)}}\,.
    \]

Algorithm~\ref{alg:iteratedCoring} iteratively increases costs and applies $\PseudoCore$. In this manner, this process iteratively removes agents that have large cost and small degree. The iterative nature of Algorithm~\ref{alg:iteratedCoring} allows to do an analysis of the properties of the resulting set $S'$.

    \begin{algorithm}[h!]
    \caption{Iterated PseudoCoring}\label{alg:iteratedCoring}
    \SetKwFunction{IteratedCoring}{IteratedCoring}
    \Input{Initial set $S^*$, costs $c_i$ for $i \in S^*$}
    For every $i \in S^*$, $c_{i,1} \gets c_i$ \;
    $S_1 \gets$ \PseudoCore{$S^*, \{c_{i,1}\}_{i \in S^*}, \funcPseudoCore$}\;
    \For{$k = 2, \dots M$}{
        \For{$i \in S_{k-1}$}{
            \If{$i\not\in C$}{
                $c_{i,k} \gets c_{i,k-1} + \Delta_{k-1}$
            }
        }
        $S_k \gets $ \PseudoCore{$S_{k-1}, \{c_{i,k}\}_{i \in S_{k-1}}, \funcPseudoCore$}
    }\
    Let $S'$ be $S_M$ after removing all agents $i$, $i\in S_M\cap C$ with  $\deg_{S_M}(i)=0$\;
    \Output{$S'$}
\end{algorithm}

For $k$, $1\leq k\leq M$ let us define
\[
    \mathcal L_k (S) := 1 - \sum_{i \in S\setminus C} \frac{c_{i,k}}{\marginal{i}{S} + 1}\,.
\]

The next lemma captures two facts. The first fact is that $\mathcal L_1 (\cdot)$ is the same as  $L(\cdot)$ up to ignoring the cost terms of agents in $C$. The second fact is that each application (moreover, each iteration) of Algorithm~\ref{alg:pseudoCore} within Algorithm~\ref{alg:iteratedCoring} improves the value of $\mathcal L_k(\cdot)$.
\begin{lemma}\label{lem:costImprovement}
    We have $L(S^*) \leq \mathcal L_1(S^*) \leq \mathcal L_1(S_1)$; and for all $k$, $1 \leq k \leq M-1$ we have $\mathcal L_{k+1}(S_{k}) \leq \mathcal L_{k+1}(S_{k+1})$.
\end{lemma}
\begin{proof}
Indeed, we have
\begin{align*}
    L(S^*)
    &=1 - \sum_{i \in S^*} \frac{c_{i}}{\marginal{i}{S} }
    =1 - \sum_{i \in S^*\setminus C} \frac{c_{i}}{\marginal{i}{S^*} }-\sum_{i \in S\cap C} \frac{c_{i}}{\marginal{i}{S^*} } \\
    &\leq 1 - \sum_{i \in S^*\setminus C} \frac{c_{i}}{\marginal{i}{S^*} }=\mathcal L_0 (S^*) \,,
\end{align*}

Consider $k$, $1 \leq k \leq M$. Consider Algorithm~\ref{alg:pseudoCore}, i.e. \PseudoCore applied to a set $S$. Let $v \in \argmin_{i \in S} \{ \marginal{i}{S}(\marginal{i}{S}+1) / c_{i,k} \}$  be about to be removed from $S$ during the current iteration of Algorithm~\ref{alg:pseudoCore}. Let us show  $\mathcal L_k(S\setminus\{v\})\geq \mathcal L_k(S)$.
\begin{align*}
    \mathcal L_k (S \setminus \{v\}) - \mathcal L_k (S)
    &= \frac{c_{v,k}}{\marginal{v}{S} + 1} - \sum_{i \in S \cap N(v)\setminus C}\inversePseudoMarginal[,k]{i}{S}{c} \\
    &\geq \frac{c_{v,k}}{\marginal{v}{S} + 1} - \sum_{i \in S \cap N(v)\setminus C}\inversePseudoMarginal[,k]{v}{S}{c} \\
    &= \frac{c_{v,k}}{\marginal{v}{S} + 1}
    \left(
        1 - \sum_{i \in S \cap N(v)\setminus C} \frac{1}{\marginal{v}{S}}
    \right) \\
    &\geq 0
\end{align*}
where the first inequality is due to the "greedy" choice of $v$ in Algorithm~\ref{alg:pseudoCore},  the second inequality is from the definition of $N(v)$ and $ \marginal{v}{S}$. This concludes the proof, since the value $\mathcal L_k(\cdot)$ is improved with each iteration when  Algorithm~\ref{alg:pseudoCore} is applied within Algorithm~\ref{alg:iteratedCoring}.
\end{proof}

\begin{lemma}\label{lem:switchingLk}
     For all $k$, $1 \leq k \leq M$ we have $\mathcal L_k (S_k) \leq \mathcal L_{k+1} (S_k) + 2\varepsilon/M$.
\end{lemma}

\begin{proof}
We want to show $\mathcal L_k (S_k) - \mathcal L_{k+1} (S_k) \leq 2\varepsilon/M$. First of all notice, that $S_k$ is the output of  Algorithm~\ref{alg:pseudoCore} applied on $S_{k-1}$ (in case of $k=1$, applied on $S^*$)  and costs $c_{i,k}$, $i\in S_{k-1}$. Recall we set $c_{i,k+1} := c_{i,k} + \Delta_k$ for all $i\in S_k\setminus C$ and so we have $c_{i,k+1} \geq \Delta_{k}$. So, by Algorithm~\ref{alg:pseudoCore}, we have
\[
    \marginal{i}{S_{k+1}}^2 + \marginal{i}{S_{k+1}}
    \geq \funcPseudoCore \cdot c_{i,k+1}
\]
and 
\[    
    \marginal{i}{S_{k+1}}
    \geq \left(\funcPseudoCore \cdot c_{i,k+1}\right)^{1/2} - 1\geq \left(\funcPseudoCore \cdot \Delta_k\right)^{1/2} - 1
\]
Thus, we have
\begin{align*}
    &\mathcal L_{k+1}(S_{k+1}) - \mathcal L_{k+2}(S_{k+1}) =
    \left( 1 - \sum_{i \in S_{k+1}\setminus C} \frac{c_{i,{k+1}}}{\marginal{i}{S_{k+1}} + 1} \right) \\
    -& \left( 1 - \sum_{i \in S_{k+1}\setminus C} \frac{c_{i,{k+2}}}{\marginal{i}{S_{k+1}} + 1} \right)
    = \sum_{i \in S_{k+1}\setminus C} \frac{c_{i,{k+2}} - c_{i,{k+1}}}{\marginal{i}{S_{k+1}} + 1} \\
    \leq& \sum_{i \in S_{k+1}\setminus C} \frac{\Delta_{k+1}}{\marginal{i}{S_{k+1}} + 1}
    \leq \sum_{i \in S_{k+1}\setminus C} \frac{\Delta_{k+1}}{\left(\funcPseudoCore \cdot \Delta_k\right)^{1/2}} \\
    \leq& \frac{n \cdot \Delta_{k+1}}{\left(\funcPseudoCore \cdot \Delta_k\right)^{1/2}}\,.
\end{align*}
Notice that by the definition of $\Delta_k$ we have
\[
\frac{\Delta_{k+1}}{\Delta_k^{1/2}}=\frac{\varepsilon\cdot \funcPseudoCore^{1/2}}{M\cdot n}\,,
\]
leading us to the desired inequality for $k\geq 2$.

For $k=1$, the inequality can be checked in the same manner. In the case $k=1$, we use that for all $i\in S_1\setminus C$ we have $c_i\geq \varepsilon/(2n)$ and so \begin{gather*}
    \marginal{i}{S_1} \geq \left(\funcPseudoCore \cdot \frac{\varepsilon}{2n} \right)^{1/2} - 1\,.
\end{gather*}
This gives us an estimate as follows
\begin{align*}
&\mathcal L_1(S_1) - \mathcal L_2(S_1) 
= \left( 1 - \sum_{i \in S_1\setminus C} \frac{c_{i,1}}{\marginal{i}{S_1} + 1} \right) - \left( 1 - \sum_{i \in S_1\setminus C} \frac{c_{i,2}}{\marginal{i}{S_1} + 1} \right) \\\\
=& \sum_{i \in S_1\setminus C} \frac{c_{i,2} - c_{i,1}}{\marginal{i}{S_1} + 1} 
\leq \sum_{i \in S_1\setminus C} \frac{\Delta_1}{\marginal{i}{S_1} + 1} \\\\
\leq& \sum_{i \in S_1\setminus C} \frac{\Delta_1}{\left(\funcPseudoCore \cdot \frac{\varepsilon}{2n} \right)^{1/2}}
\leq  2\frac{n^{3/2} }{\left(\funcPseudoCore \cdot \varepsilon\right)^{1/2}}\cdot \Delta_1  \\\\
=& 2\frac{n^{3/2} }{\left(\funcPseudoCore \cdot \varepsilon\right)^{1/2}} \cdot \frac{\varepsilon}{M} \cdot \frac{(\funcPseudoCore \cdot \varepsilon)^{1/2}}{n^{3/2}} 
\leq 2\frac{\varepsilon}{M}\,.
\end{align*}

\end{proof}

Finally we show that our actual $L(\cdot)$ value is not too far off from our relaxed value at the end of this process.

\begin{lemma}\label{lem:switchingL}
  We have  $\mathcal L_M(S_M) \leq L(S^*) + 2\varepsilon$.
\end{lemma}

\begin{proof}
Note that Algorithm~\ref{alg:iteratedCoring} does not decrease the costs of agents, so $c_{i,M}\geq c_{i}$ for each $i\in S_M$. Note also that $S'$ is obtained from $S_M$ by excluding all agents $i\in S_M\cap C$ such that $\deg_{S_M}(i)=0$. So we have
\begin{align*}
    &\mathcal L_M(S_M) - L(S') \\
    =& \left( 1 - \sum_{i \in S_M\setminus C} \frac{c_{i,M}}{\marginal{i}{S_M}+1} \right) - \left( 1 - \sum_{i \in S_M\setminus C} \frac{c_{i}}{\marginal{i}{S_M}}- \sum_{i \in S'\cap C} \frac{c_{i}}{\marginal{i}{S_M}}\right) \\
    \leq& \sum_{i \in S_M\setminus C} \frac{c_{i}}{\marginal{i}{S_M}}-\sum_{i \in S_M\setminus C} \frac{c_{i}}{\marginal{i}{S_M}+1}+|S\cap C|\frac{\varepsilon}{2n} \\
    \leq& \sum_{i \in S_M\setminus C} \frac{c_i + \marginal{i}{S_M} \cdot (c_i - c_{i,M})} {\marginal{i}{S_M} \cdot (\marginal{i}{S_M}+1)} +\frac{\varepsilon}{2} \\
    \leq&_\star \sum_{i \in S_M\setminus C} \frac{c_i} {\marginal{i}{S_M} \cdot (\marginal{i}{S_M}+1)} + \frac{\varepsilon}{2} \\
    \leq&_{\star\star} n \cdot \left( \frac{36e}{\varepsilon^5 \cdot (1 - \varepsilon)^2} \cdot \frac{\log_2 \ln(n)}{n} \right)^2 +\frac{\varepsilon}{2} \leq 2\varepsilon
\end{align*}
where the  inequality $\star$ follows from $c_i \leq c_{i,M}$ for all $i\in S_M\setminus C$, and the  inequality $\star\star$ follows from Lemma~\ref{lem:highDegrees} and Observation~\ref{obs:averageCost}, and the final inequality is true for large enough $n$.
\end{proof}

Let us first state the corollary of Lemma~\ref{lem:costImprovement}, Lemma~\ref{lem:switchingLk} and Lemma~\ref{lem:switchingL}.
\begin{corollary}
    We have $L(S^*) \leq L(S') + 4\varepsilon$.
\end{corollary}
\begin{proof} This follows directly from the inequalities in Lemma~\ref{lem:costImprovement}, Lemma~\ref{lem:switchingLk}, and Lemma~\ref{lem:switchingL}.
    \begin{align*}
        L(S^*) \leq \mathcal L_1(S^*)
        &\leq \mathcal L_1(S_1) \leq \mathcal L_2(S_1) + 2\varepsilon/M \\
        &\leq \mathcal L_2(S_2) \leq \mathcal L_3(S_2) + 4\varepsilon/M \leq \ldots  \\
        &\leq \mathcal L_M(S_M) + 2\varepsilon \\
        &\leq L(S') + 4\varepsilon
    \end{align*}
\end{proof}

\begin{lemma}\label{lem:highDegrees}
    For all $i \in S'\setminus C$, we have
    \begin{align*}
        \marginal{i}{S'}\geq \frac{\varepsilon^5 \cdot (1 - \varepsilon)^2}{36e} \cdot \frac{n}{\log_2 \ln(n)}
        \,, \,  \text{ and } \\
        \frac{\marginal{i}{S'}(\marginal{i}{S'}+1)}{c_i} \geq \funcPseudoCore \,.
    \end{align*}
\end{lemma}
\begin{proof}
    The second property is immediate from the application of Algorithm~\ref{alg:pseudoCore} in Algorithm~\ref{alg:iteratedCoring}. Consider the first property. We can apply our degree bound and our bound on $\Delta_{M-1}$ to $S_M$ to get
\begin{align*}
    &\marginal{i}{S_M}
    \geq \left(\funcPseudoCore \cdot c_{i,M-1}\right)^{1/2} - 1 \\
    \geq& \frac{\varepsilon^{\left(1 - \frac{1}{2^M}\right)} \cdot \left(\funcPseudoCore\right)^{\left(1 - \frac{1}{2^M}\right)}}  {M^{\left(1 - \frac{1}{2^{M-1}}\right)} \cdot n^{\left(1 - \frac{1}{2^M}\right)}} - 1 \\
    \geq& \varepsilon \cdot \frac{\funcPseudoCore}{n} \cdot  \frac{1}{M}\left(\frac{\funcPseudoCore}{n}\right)^{-\frac{1}{2^M}} - 1 \,.
\end{align*}
So by running Algorithm~\ref{alg:iteratedCoring} for $M = \log_2 \ln(\gamma(n+6/\varepsilon))$ iterations, for all $i\in S_M\setminus C$ we get 
 \[
        \marginal{i}{S'}
        \geq \frac{\varepsilon}{e} \cdot \frac{\gamma(n+6/\varepsilon)}{\log_2 \ln(\gamma(n+6/\varepsilon))} - 1
        \geq \frac{\varepsilon^5 \cdot (1 - \varepsilon)^2 \ln(2)}{36e} \cdot \frac{n}{\ln \ln(n)} \,,
    \]
which gives the desired bound by substituting and simplifying.
\end{proof}

Now let us finally observe that our set $S_M$ retains approximately optimal $R(\cdot)$ value.
We observe that if the pseudo-coring Algorithm~\ref{alg:pseudoCore} only removes elements $i \in S$ with $\marginal{i}{S} < \funcPseudoCoreInternal$ then in the course of our iterated coring Algorithm~\ref{alg:iteratedCoring} we remove at most $\funcPseudoCoreInternal \cdot n = \frac{\varepsilon}{6}\cdot n^2$ value from $f(S^*)$, which reduces $R(\cdot)$ by less than $\varepsilon$. However, the pseudo-coring algorithm may remove elements of degree up to $\maxMarginal$, if the costs of those elements are high enough. We bound the number of such elements the algorithm may remove.
\begin{lemma}\label{lem:iteratedCoreHighValue}
    We have $R(S^*) \leq R(S') + \varepsilon$, $R(S') \geq \frac{\varepsilon}{3}$ and
    \[
        |\{i \in S' \mid \marginal{i}{S'} \geq \frac{\varepsilon}{6} \cdot \maxMarginal\}| \geq \frac{\varepsilon}{6} n\,.
    \]
\end{lemma}
\begin{proof}
Let us consider $\widetilde S$ as in Lemma~\ref{lem:optimalSize} and let us show that $|\widetilde S\setminus S_M|\leq \varepsilon n/6$, implying $|\widetilde S\setminus S'|\leq \varepsilon n/6$. This would finish the proof, since $|\widetilde S|\geq \varepsilon n/3$ and for all $i\in \widetilde S$ we have $\deg_{\widetilde S}(i)\geq \varepsilon n/3$. Thus, we need to show that at most $\varepsilon n/6$ agents from $\widetilde S$ are removed during Algorithm~\ref{alg:iteratedCoring}.

     Let us consider an iteration of Algorithm~\ref{alg:iteratedCoring} at the moment when the total number of removed agents in $\widetilde S$ is at most $\varepsilon n/6$. Consider iteration $k$, $1\leq k\leq M$ of Algorithm~\ref{alg:iteratedCoring}, i.e. the iteration when Algorithm~\ref{alg:pseudoCore} is called with costs $c_{i,k}$. Let us assume that during the call to Algorithm~\ref{alg:pseudoCore} a vertex $v\in \widetilde S$ is removed from $\bar{S}$ in Algorithm~\ref{alg:pseudoCore} . Thus,  we have
     \[
        \pseudoMarginal[,k]{v}{\bar S}{c} < \funcPseudoCore \,,
     \]
     but also $\marginal{v}{\bar S} \geq  \funcPseudoCoreInternal=\varepsilon n/6$, where $\bar S$ is the set in our algorithm at time of removal. We have
    \begin{align*}
        c_{v,k} &= c_v + \sum_{j=1}^k \Delta_j
        \leq c_v + \frac{\varepsilon \cdot \funcPseudoCore}{n^2} \cdot \left(\frac{\funcPseudoCore}{n}\right)^{-\frac{1}{2^k}} \\
        &\leq c_v + \frac{\varepsilon \cdot \funcPseudoCore}{e \cdot n^2} \,,
    \end{align*}
    where the first inequality is due to the definition of $\Delta_k$, and the second is from the fact that $k \leq M = \log_2 \ln(\gamma(n+6/\varepsilon))$.
    Now we also have
    \begin{gather*}
        \funcPseudoCore > \pseudoMarginal[,k]{v}{\bar S}{c} \geq \frac{\funcPseudoCoreInternal \cdot (\funcPseudoCoreInternal + 1)}{c_{i,k}} \,,
    \end{gather*}
    since  $v$ is removed in Algorithm~\ref{alg:pseudoCore} and at the moment of removal $\marginal{v}{\bar{S}} \geq \funcPseudoCoreInternal$. Rearranging and putting this together yields
    \[
        c_v \geq \frac{\funcPseudoCoreInternal \cdot (\funcPseudoCoreInternal + 1)}{\funcPseudoCore} - \frac{\varepsilon \cdot \funcPseudoCore}{e \cdot n^2}
        = \frac{n}{\varepsilon^2 \cdot n} - \frac{\varepsilon^2}{6 e n} \left(\frac{\varepsilon}{3}\maxMarginal - 1\right)
        \geq \frac{(1-\varepsilon)}{\varepsilon^2}\,.
    \]

Hence, for every agent $v\in\widetilde S$ that was removed in Algorithm~\ref{alg:pseudoCore} before in total more than $\varepsilon n/6$ agents were removed, we have
\[c_v\geq (1-\varepsilon)/(\varepsilon^2)\,.\]
By Observation~\ref{obs:averageCost}, we have $\sum_{i\in \widetilde S}c_i<n$. This shows that Algorithm~\ref{alg:pseudoCore} removes at most $\varepsilon^2 \cdot n/(1-\varepsilon)$ vertices from $\widetilde S$. Thus, as desired $|\widetilde S\setminus S_M|\leq n/(1-\varepsilon) \leq \varepsilon n/6$. 

To finish the proof notice that
\[ R(S^*)-\varepsilon/3\leq  R(S')\leq R(S')+\frac{\varepsilon n}{6}\cdot n\cdot\frac{1}{\binom{n}{2}}=R(S')\varepsilon/3\,. \]

\end{proof}

\end{document}